\documentclass[aps,prd,reprint,superscriptaddress,longbibliography,showpacs,nofootinbib]{revtex4-2}

\usepackage{graphicx}
\usepackage{svg}
\usepackage{dcolumn}
\usepackage{bm}

\usepackage[figurename=Figure]{caption}
\usepackage{ragged2e}
\DeclareCaptionJustification{plain}{\justifying}
\usepackage[percent]{overpic}
\usepackage{float}    
\usepackage{verbatim} 
\usepackage{amsmath}  
\usepackage{bbold}
\usepackage{upgreek} 
\usepackage{amssymb}  
\usepackage{svg}
\usepackage{enumitem} 
\usepackage{latexsym,epsfig}
\usepackage{subcaption}
\usepackage{stackrel}
\usepackage{mathtools}
\usepackage{amsthm}
\newtheorem{proposition}{Proposition}
\newtheorem{definition}{Definition}

\newtheorem{theorem}{Theorem}

\usepackage{placeins} 

\usepackage[colorlinks=true,breaklinks=true,allcolors=blue]{hyperref}
\usepackage[capitalize]{cleveref}
\usepackage{lipsum}
\usepackage{dsfont}
\usepackage{coffeestains}
\newtheorem*{remark}{Remark}

\usepackage[scr=boondox]{mathalfa}
 
\usepackage[scr=boondox]{mathalfa}

\newcommand{\be}{\begin{equation}}
	\newcommand{\ee}{\end{equation}}
\newcommand{\bea}{\begin{equation}\begin{aligned}}
		\newcommand{\eea}{\end{aligned}\end{equation}}
\newcommand{\ben}{\begin{enumerate}}
	\newcommand{\een}{\end{enumerate}}

\DeclareDocumentCommand{\nint}{ O{} O{} m }{\ensuremath{ \int_{\mbox{\scriptsize $#1$}}^{\mbox{\scriptsize$#2$}}\!\!\! \mbox{\small $\,\mathrm{d}#3$\! }}}

\usepackage{tikz,bm} 
\usepackage{circuitikz}
\usetikzlibrary{decorations.markings}
\usetikzlibrary{decorations.pathreplacing,calligraphy}

\definecolor{mycolor}{rgb}{1,0.2,0.3}
\definecolor{brightgreen}{rgb}{0.4, 1.0, 0.0}
\definecolor{britishracinggreen}{rgb}{0.0, 0.26, 0.15}
\definecolor{cadmiumgreen}{rgb}{0.0, 0.42, 0.24}
\definecolor{ceruleanblue}{rgb}{0.16, 0.32, 0.75}
\definecolor{darkelectricblue}{rgb}{0.33, 0.41, 0.47}
\definecolor{darkpowderblue}{rgb}{0.0, 0.2, 0.6}
\definecolor{darktangerine}{rgb}{1.0, 0.66, 0.07}
\definecolor{emerald}{rgb}{0.31, 0.78, 0.47}
\definecolor{palatinatepurple}{rgb}{0.41, 0.16, 0.38}
\definecolor{pastelviolet}{rgb}{0.8, 0.6, 0.79}

\begin{document}

	\title{Geometric Decoherence Time in Lindbladian Dynamics}
	\author{Rishabh Jha}
	\email{rishabh.jha@usc.edu}
	\affiliation{%
		Department of Physics and Astronomy, University of Southern California, Los Angeles, CA 90089-0484, USA
	}
	\author{Stephan Haas}
	\email{shaas@usc.edu}
	\affiliation{%
		Department of Physics and Astronomy, University of Southern California, Los Angeles, CA 90089-0484, USA
	}
	\author{Abhinav Prem}
	\email{aprem@bard.edu}
	\affiliation{Physics Program, Bard College, 30 Campus Road, Annandale-on-Hudson, NY 12504, USA}
	%
	
	
\begin{abstract}

The onset of decoherence in open many-body systems lacks a dynamical timescale grounded in the loss of bipartite entanglement. Here, we introduce the \textit{geometric decoherence time}, defined as the earliest moment the monotone relation between logarithmic negativity and R\'{e}nyi-$\tfrac{1}{2}$ entropy---exactly equal across any bipartition for pure states---breaks down under open-system evolution, signaling entropy growth without accompanying entanglement growth. We establish this criterion in both single-particle Gaussian dynamics and many-body Lindbladian evolution. We show that quantum mutual information provides a complementary long-time diagnostic: its asymptotic vanishing is equivalent to factorization of the steady state across the bipartition, a condition strictly stronger than separability, and whenever a product steady state is approached exponentially in trace norm, negativity and mutual information share the same decay rate. In the presence of a strong symmetry, this tracking can fail---residual classical correlations can survive after entanglement has vanished. In the Kitaev chain with balanced gain and loss, we derive a closed-form solution and show that the topological phase sustains longer coherence times than the trivial phase at identical dissipation, with a local minimum at the chiral-symmetric point. In the interacting XXZ chain, exact many-body evolution shows that local $Z$-dephasing preserves residual classical correlations, whereas gain and loss restore the mutual-information tracking of negativity. Our results establish the geometric decoherence time as a dynamical scale tracking the onset of decoherence.
	\end{abstract}

	\maketitle
	



	\section{Introduction}
	\label{sec:introduction}
	
The loss of quantum coherence through coupling to an environment is a central feature of open system dynamics, with implications ranging from the emergence of classicality to the limitations of quantum information storage and processing~\cite{Zurek2003Decoherence,Schlosshauer2005Decoherence,Breuer2002OpenQS,Rivas2012OpenQSIntro}. In the standard open-system framework, decoherence is encoded in the spectral properties of quantum dynamical semigroups and master equations~\cite{Breuer2002OpenQS,Breuer2002ConceptsOpenQS,Rivas2012OpenQSIntro} where, in practice, one invokes relaxation or dephasing time scales extracted from single-time observables. However, for extended many-body systems, coherent entanglement growth coexists and competes with noise in the approach to the non-equilibrium steady state (NESS), rendering the usual characterizations of the onset of decoherence opaque. As such, a crisp notion of a decoherence timescale---one which is explicitly dynamical, sensitive to genuine quantum correlations, and robust across microscopic descriptions---remains lacking.

Entanglement provides the natural language for tracking quantum many-body dynamics, in both unitary and open-system settings~\cite{Amico2008EntanglementManyBody,Laflorencie2016EntanglementReview}. In closed systems, the entanglement entropy and its R\'enyi generalizations display universal scaling at criticality~\cite{Calabrese2004EntanglementQFT,Calabrese2009CFTreview,CalabreseCardy2009JSMTE} and, in integrable systems, their dynamics after global quenches is well characterized via a quasiparticle picture~\cite{AlbaCalabrese2017PNAS,AlbaCalabrese2018SciPost}. For mixed states, which naturally arise under open system dynamics, the entanglement entropy alone is insufficient; instead, the logarithmic negativity~\cite{Peres1996PPT,VidalWerner2002Negativity} (a computable monotone under positive-partial transpose preserving operations~\cite{Plenio2005LogNeg}) is a standard mixed-state entanglement measure~\cite{Calabrese2012NegativityQFT,NegativityDimension2013,EltschkaSiewert2014QuantifyingEnt,Hillery2024NegativityBounds,Coser2014NegativityQuench,Sherman2016NegativityFiniteT,lu2020detect,sang2021ent}. For free fermions, the fermionic entanglement negativity~\cite{Shapourian2017FermionicNegativity,shapourian2019,shapourian2019neg,Shapourian2021NegativitySpectrum} can be accessed via the correlation matrix and has been extensively investigated for non-interacting dissipative chains~\cite{alba2022,AlbaCarollo2022OpenNegativity,murciano2023,caceffo2023,caceffo2026}. Despite much progress in characterizing entanglement in open system dynamics, what is still missing is a notion of decoherence time based on the entanglement negativity that is defined via the loss of bipartite entanglement, rather than indirectly through spectra, transport coefficients, or late-time relaxation fits.

This gap is especially acute in open many-body systems. Exact solutions of boundary-driven and dephased spin chains~\cite{Prosen2011OpenXXZ,Znidaric2010DephasingXX,Yamanaka2023OpenXX,turkeshi2021}, as well as studies of monitored quantum circuits~\cite{Skinner2019MIPT,GullansHuse2019Purification,Jian2021MIPTSYK}, have revealed rich non-equilibrium physics, but the relevant timescales here are usually extracted from Liouvillian gaps, transport coefficients, or measurement rates. Experimentally, randomized measurements and classical-shadow tomography provide access to moments of the partial transpose and thus to the negativity~\cite{Elben2020MixedState,Zhou2020Negativity,Jeng2019MBEDistillation}, while quantum-gas microscopes provide R\'enyi entropies, mutual information, and higher-order correlators in lattice geometries~\cite{Bakr2009QuantumGasMicroscope,Gluza2020QuantumReadout}. It is thus timely to ask for a definition of decoherence time that is both operational and tied to bipartite correlations.


In this work, we provide a conceptually simple picture for the onset of decoherence, which we argue holds for both single-particle (i.e., Gaussian) and interacting Lindbladian dynamics. Our starting point is the exact identity between the logarithmic negativity ($E_N$) and the R\'enyi-$1/2$ entropy ($S_{1/2}$) that holds across any bipartition for a pure state~\cite{Calabrese2012NegativityQFT,Calabrese2004EntanglementQFT} (more generally, an exact universal relation between these two quantities has also been established in the hydrodynamic limit of integrable systems~\cite{AlbaCalabrese2019EPL,Bertini2022NegativityMI,CaceffoAlba2023TMI}). Under unitary dynamics from a pure state, the trajectory $(S_{1/2}(t),E_N(t))$ is pinned to the diagonal of the entropy--negativity plane. Open-system dynamics typically makes the bipartite state mixed and allows this relation to fail. We thus define the \textit{geometric decoherence time} $\tau_{\mathrm{d}}^{\mathrm{g}}$ as the earliest time at which the monotone relation between $S_{1/2}$ and negativity is lost, signaling the onset of entropy growth without accompanying entanglement growth. We provide analytical and numerical evidence that this dynamical timescale captures the onset of decoherence, and is quantitatively distinct from the peak time of the negativity, $\tau_{\mathrm{d}}^{\mathrm{peak}}$, a natural proxy introduced in earlier studies of negativity dynamics~\cite{Coser2014NegativityQuench,Sherman2016NegativityFiniteT,AlbaCarollo2022OpenNegativity}. 

 
We further show that the quantum mutual information plays a complementary role: it measures total correlations rather than entanglement and satisfies the exact criterion $I(A{:}B)=0$ if and only if $\rho_{AB}=\rho_A\otimes\rho_B$~\cite{NielsenChuangBook}. Hence, its vanishing detects factorization across the cut, a stronger condition than separability. Unital primitive dynamics gives one simple sufficient route to such a product NESS, since the unique stationary state is then maximally mixed and factorized, but the result applies to any product NESS. We show that factorization across the cut controls the joint late-time behavior of $I(A{:}B)$ and $E_N$: when the NESS factorizes across the bipartition, both diagnostics vanish asymptotically. We rigorously establish that if the approach to a product NESS is exponential with rate $\lambda$, then $I(A{:}B)$ and $E_N$ decay with the same asymptotic exponent $\lambda$, up to the standard entropy-continuity correction for mutual information. Beyond this asymptotic statement, we empirically find that, in the absence of a conservation law or kinetic constraint, mutual information tracks the peak, trough, and decay times of negativity throughout the Lindbladian dynamics. This tracking fails in the presence of a strong symmetry for the simple reason that negativity can vanish while mutual information remains finite, reflecting classical correlations that survive due to the corresponding conservation law.

We demonstrate these results in two settings. First, for the Kitaev chain with gain and loss, we utilize the standard third quantization approach to derive a closed-form expression for the correlation-matrix dynamics. In the case of balanced gain and loss, this solution us to map $\tau_{\mathrm{d}}^{\mathrm{g}}$ across the phase diagram as a function of varying chemical potential and the gain/loss rate. The topological phase sustains longer coherence times than the trivial phase at any given dissipation rate, with a local minimum at the chiral-symmetric point, showing that the closed system topology influences the decoherence time. Second, for the interacting XXZ chain, exact many-body Lindblad evolution shows that local dephasing (which preserves the strong symmetry associated with conservation of total magnetization) can leave a non-zero mutual-information plateau after negativity has vanished. Adding gain and loss breaks this strong symmetry, which restores asymptotic factorization.

This paper is organized as follows. In Sec.~\ref{sec:theory} we develop the information-theoretic framework, introducing the entropy--negativity trajectory $\Gamma(t)$, the geometric definition of $\tau_{\mathrm{d}}^{\mathrm{g}}$, and the role of mutual information as a diagnostic of asymptotic factorization. We establish that product factorization controls the joint late-time behavior of $I(A{:}B)$ and $E_N$, and that, when the approach to a product NESS is exponential, both diagnostics decay with the same asymptotic exponent. In Sec.~\ref{sec:single_body} we apply this framework to the exactly solvable open Kitaev chain, derive the correlation-matrix dynamics for balanced and imbalanced gain/loss, compare $\tau_{\mathrm{d}}^{\mathrm{g}}$ and $\tau_{\mathrm{d}}^{\mathrm{peak}}$ across the topological phase diagram in the chemical potential--dissipation $(\mu,\gamma)$ plane, and perform finite-size and finite-cut scaling. In Sec.~\ref{sec:many_body}, we turn to the interacting XXZ spin chain under local dephasing and under gain/loss dissipation, using exact many-body Lindblad evolution to probe the geometric criterion beyond Gaussian evolution. Sec.~\ref{sec:conclusion} summarizes our results and outlines future directions.

\section{Geometric Decoherence Time: Definitions and Diagnostics}
\label{sec:theory}

For any bipartite pure state $\rho_{AB}=|\psi\rangle\langle\psi|$ and any bipartition $A|B$, the logarithmic negativity equals the R\'enyi-$\tfrac{1}{2}$ entropy of either reduced state,
\begin{equation}
	E_N(\rho_{AB}) = S_{1/2}(\rho_A) = S_{1/2}(\rho_B).
	\label{eq:pure_identity}
\end{equation}
This identity follows from the Schmidt decomposition alone and is implicit in Refs.~\cite{VidalWerner2002Negativity,Calabrese2012NegativityQFT,Calabrese2004EntanglementQFT}; a self-contained proof is given in Appendix~\ref{app:pure_state_proof} for completeness. It holds throughout any unitary dynamics from a pure initial state, for any spatial dimension, any degree of integrability, and arbitrarily far from equilibrium. Under unitary dynamics the joint trajectory $\Gamma(t)=\bigl(S_{1/2}(t),\mathcal{N}(t)\bigr)$ in the entropy--negativity plane is pinned to the diagonal $\mathcal{N}=S_{1/2}$ for all time. Open-system evolution drives the bipartite state into a mixed state, lifts the trajectory off the diagonal, and eventually breaks the monotone relation between the two quantities. We define the geometric decoherence time $\tau_{\mathrm{d}}^{\mathrm{g}}$ as the earliest moment this monotone relation fails. All logarithms are base $2$ throughout.

\subsection{Entanglement Diagnostics}
\label{subsec:definitions}

We consider open-system dynamics governed by the Gorini--Kossakowski--Lindblad--Sudarshan (GKLS) master equation~\cite{Breuer2002OpenQS,Rivas2012OpenQSIntro}
\begin{equation}
	\frac{d\rho}{dt}
	=
	-i[H,\rho]
	+
	\sum_\mu
	\!\left(
	L_\mu \rho L_\mu^\dagger
	-\tfrac{1}{2}\{L_\mu^\dagger L_\mu,\rho\}
	\right),
	\label{eq:lindblad_general}
\end{equation}
where $H$ generates unitary dynamics and $L_\mu$ are jump operators encoding the coupling to the environment. In the exact many-body setting, Eq.~\eqref{eq:lindblad_general} is solved for the full density matrix. For quadratic fermionic Hamiltonians with linear jump operators, the third-quantization formalism~\cite{Prosen2008ThirdQuantization,Prosen2010ThirdQuantization} closes the hierarchy at the level of the two-point correlation matrix. For a bipartite state $\rho_{AB}$ with $B=A^c$, the reduced density matrices are $\rho_A=\mathrm{Tr}_B\rho_{AB}$ and $\rho_B=\mathrm{Tr}_A\rho_{AB}$. We use the von Neumann entropy $S(\rho)=-\mathrm{Tr}(\rho\log\rho)$ and the R\'enyi entropy
\begin{equation}
	S_\alpha(\rho)
	= \frac{1}{1-\alpha}\log\mathrm{Tr}(\rho^\alpha),
	\qquad \alpha>0,\;\alpha\neq 1,
	\label{eq:entropies}
\end{equation}
with $S_{1/2}(\rho)=2\log\mathrm{Tr}\sqrt{\rho}$ the case relevant to Eq.~\eqref{eq:pure_identity}. The quantum mutual information is
\begin{equation}
	I(A{:}B) = S(\rho_A) + S(\rho_B) - S(\rho_{AB}),
	\label{eq:MI_def}
\end{equation}
and the logarithmic negativity is~\cite{VidalWerner2002Negativity,Plenio2005LogNeg}
\begin{equation}
	E_N(\rho_{AB}) = \log\!\left\|\rho_{AB}^{T_B}\right\|_1,
	\label{eq:LN_def}
\end{equation}
where $T_B$ denotes partial transpose on $B$ and $\|X\|_1=\mathrm{Tr}\sqrt{X^\dagger X}$ is the trace norm. The two quantities measure distinct aspects of bipartite correlations: $E_N$ is an entanglement monotone and vanishes on all separable states, while $I(A{:}B)$ quantifies total correlations and vanishes only when $\rho_{AB}$ factorizes. The classically correlated state $\rho_{AB}=\frac{1}{2}(|00\rangle\langle 00|+|11\rangle\langle 11|)$ illustrates the gap: it has $E_N=0$ yet $I(A{:}B)=1$. Despite this difference, both quantities are accessible from the same bipartite trajectory $\Gamma(t)$ and, as established in Sec.~\ref{subsec:MI_asymptotics}, they share the same asymptotic decay exponent whenever the steady state factorizes.

For fermionic Gaussian states, the partial transpose does not preserve Gaussianity, so $E_N$ is not directly accessible from the correlation matrix~\cite{Shapourian2017FermionicNegativity}. Let $c_j$, $c_j^\dagger$ ($j=1,\ldots,L$) denote fermionic operators satisfying $\{c_i,c_j^\dagger\}=\delta_{ij}$. The two-point correlation matrix $(C_{AB})_{ij}=\langle c_j^\dagger c_i\rangle$ determines the shifted matrix $G_{AB}=2C_{AB}-\mathds{1}$, written in block form as
\begin{equation}
	G=
	\begin{pmatrix}
		G_{AA} & G_{AB}^{\phantom{\dagger}}\\
		G_{BA} & G_{BB}
	\end{pmatrix}.
	\label{eq:GA_blocks}
\end{equation}
The fermionic partial-transpose construction~\cite{Shapourian2017FermionicNegativity,shapourian2019} introduces auxiliary matrices
\begin{equation}
	G_\pm=
	\begin{pmatrix}
		G_{AA} & \pm i\,G_{AB}^{\phantom{\dagger}}\\
		\pm i\,G_{BA} & -G_{BB}
	\end{pmatrix},
	\label{eq:Gpm_def}
\end{equation}
and defines $M=\mathds{1}+G_+G_-$ and $\widetilde{G}^{\,T_B}=\frac{1}{2}[\mathds{1}-M^{-1}(G_++G_-)]$. Denoting by $\{\nu_j\}$ and $\{\widetilde{\nu}_j\}$ the eigenvalues of $C_{AB}$ and $\widetilde{G}^{\,T_B}$, the fermionic Gaussian negativity is~\cite{AlbaCarollo2022OpenNegativity}
\begin{equation}
	\begin{aligned}
	\mathcal{E}_F(\rho_{AB})
	=& \, 
	\frac{1}{2}\sum_j \log  \!\left(
	\sqrt{\widetilde{\nu}_j}+\sqrt{1-\widetilde{\nu}_j}
	\right) \\
	&+
	\frac{1}{4}\sum_j\log_2\!\left[
	\nu_j^2+(1-\nu_j)^2
	\right] \, .
	\end{aligned}
	\label{eq:fermionic_negativity_spectrum}
\end{equation}
We reserve $E_N$ for the exact many-body logarithmic negativity and $\mathcal{E}_F$ for its Gaussian counterpart. Both satisfy the pure-state identity in Eq.~\eqref{eq:pure_identity} when $\gamma=0$, so the geometric decoherence criterion is well-defined in either case. The unified notation
\begin{equation}
	\mathcal{N}(t)=
	\begin{cases}
		E_N(t), & \text{exact many-body,}\\
		\mathcal{E}_F(t), & \text{Gaussian,}
	\end{cases}
\end{equation} 
applies whenever the discussion holds in both settings; the geometric decoherence time $\tau_{\mathrm{d}}^{\mathrm{g}}$ is defined identically in terms of $\mathcal{N}$ in either case.

\subsection{Geometric Decoherence Time}
\label{subsec:geom_dec_time}

As stated above, under unitary dynamics from a pure initial state $\Gamma(t)$ is pinned to the diagonal; under Lindbladian dynamics, it bends away from the diagonal, with the earliest moment this occurs defined as the geometric decoherence time. 
\begin{definition}[Geometric decoherence time]
	\label{def:geom_dec}
Assume $t\mapsto(S_{1/2}(t),\mathcal{N}(t))$ is continuous. The \emph{geometric decoherence time} $\tau_{\mathrm{d}}^{\mathrm{g}}$ is the infimum of times $t>0$ for which there exist $t_1<t_2\le t$ such that
	\begin{equation}
		S_{1/2}(t_2)\ge S_{1/2}(t_1)
		\quad\text{but}\quad
		\mathcal{N}(t_2)<\mathcal{N}(t_1).
		\label{eq:geom_break}
	\end{equation}
When the trajectory is differentiable and the early-time branch over $S_{1/2}$ is single-valued, this is equivalently
	\begin{equation}
		\tau_{\mathrm{d}}^{\mathrm{g}}
		=
		\inf\!\left\{t>0:\,\dot{\mathcal{N}}(t)<0
		\;\text{while}\;\dot{S}_{1/2}(t)\ge 0\right\}.
		\label{eq:geom_break_local}
	\end{equation}
\end{definition}
Here, the label ``geometric'' refers to the shape of the parametric curve $\Gamma(t)\subset\mathbb{R}^2$ and this definition captures the onset of entanglement loss in the presence of continuing entropy growth. A natural comparison scale, motivated by earlier studies on negativity dynamics after quenches and under dissipation~\cite{Coser2014NegativityQuench,Sherman2016NegativityFiniteT,AlbaCarollo2022OpenNegativity}, is the time at which $\mathcal{N}$ attains its global maximum,
\begin{equation}
	\tau_{\mathrm{d}}^{\mathrm{peak}}
	= \arg\max_{t\ge 0}\,\mathcal{N}(t).
	\label{eq:peak_dec}
\end{equation}
The two scales are logically distinct: the peak of $\mathcal{N}$ identifies the extremum of a single time trace, while $\tau_{\mathrm{d}}^{\mathrm{g}}$ identifies the earliest moment the joint trajectory loses monotonicity. We discuss the distinction between these two diagnostics in Sec.~\ref{sec:single_body} and show that the geometric criterion is a sharper diagnostic for the onset of decoherence.

\subsection{Mutual Information and Asymptotic Factorization}
\label{subsec:MI_asymptotics}

Evaluating the logarithmic negativity requires either full state tomography or replica methods. In contrast, the mutual information $I(A{:}B)$ only requires the reduced density matrices on $A$ and $B$, is directly accessible from subsystem entropies, and satisfies the bound $E_N\le S_{1/2}(\rho_A)$~\cite{Plenio2005LogNeg}. This motivates asking whether $I(A{:}B)$ tracks the same dynamical timescales as $\mathcal{N}$. The rigorous answer for the long-time behavior follows from the relative-entropy representation~\cite{NielsenChuangBook}
\begin{equation}
	I(A{:}B)
	= D\!\left(\rho_{AB}\,\middle\|\,\rho_A\otimes\rho_B\right)
	\ge 0,
	\label{eq:MI_relative_entropy}
\end{equation}
where $D(\rho\|\sigma)=\mathrm{Tr}(\rho\log\rho)-\mathrm{Tr}(\rho\log\sigma)$,
which gives the precise sense in which $I(A{:}B)$ measures total correlations: it satisfies
\begin{equation}
	I(A{:}B)=0
	\;\Longleftrightarrow\;
	\rho_{AB}=\rho_A\otimes\rho_B.
	\label{eq:MI_iff_product}
\end{equation}
The mutual information vanishes only when the state factorizes, not merely when entanglement vanishes.

\begin{proposition}[Product factorization controls the joint asymptotics]
	\label{thm:factorization_joint}
Suppose the Lindblad evolution admits a trace-norm limit
	$\rho(t)\to\rho_\infty$ as $t\to\infty$. Then:
	\begin{enumerate}[label=(\roman*)]
		\item $\rho_\infty=\rho_{\infty,A}\otimes\rho_{\infty,B}$
		if and only if $\lim_{t\to\infty}I(A{:}B;t)=0$.
		\item If $\rho_\infty$ factorizes, then
		$\lim_{t\to\infty}E_N(t)=0$.
	\end{enumerate}
\end{proposition}
The proof, given in Appendix~\ref{app:factorization_proof}, uses only Eq.~\eqref{eq:MI_iff_product} and continuity of $I$ in trace norm on finite-dimensional Hilbert spaces. Proposition~\ref{thm:factorization_joint} establishes that the asymptotic vanishing of $I(A{:}B)$ is the necessary and sufficient signature of product factorization of the NESS, independent of whether the NESS is maximally mixed, primitive, or unital.

\textbf{Remark} [Sufficient condition and its limits]. A GKLS evolution is called \emph{primitive} if it admits a unique steady state, and \emph{unital} if $\mathds{1}/d$ is a steady state~\cite{Breuer2002OpenQS,Wolf2012QChannels}; when both hold, the unique steady state $\mathds{1}/d$ factorizes as $(\mathds{1}_A/d_A)\otimes(\mathds{1}_B/d_B)$, satisfying condition~(i) and forcing $I(A{:}B)$ and $E_N$ to vanish asymptotically (see Appendix~\ref{app:unital_primitive}). This is sufficient but not necessary: any product NESS, not only the maximally mixed one, satisfies condition~(i). Conversely, if a strong symmetry prevents the NESS from factorizing, condition~(i) fails and $I$ need not vanish, as demonstrated in Sec.~\ref{sec:many_body}. The results of this section require neither primitivity nor unitality.

\begin{theorem}[Shared asymptotic decay exponent]
	\label{cor:product_decay}
Under the assumptions of Proposition~\ref{thm:factorization_joint}, suppose further that $\|\rho(t)-\rho_\infty\|_1\le Ce^{-\lambda t}$ for constants $C,\lambda>0$ and all sufficiently large $t$. Then
	\begin{equation}
		I(A{:}B;t)=O\!\left(te^{-\lambda t}\right),
		\qquad
		E_N(t)=O\!\left(e^{-\lambda t}\right).
		\label{eq:quantitative_product_rates}
	\end{equation}
\end{theorem}
Both $I$ and $E_N$ therefore share the asymptotic decay exponent $\lambda$ near a product fixed point, with $I$ acquiring at most a logarithmic prefactor from the Fannes--Audenaert inequality~\cite{Fannes1973Dec,Audenaert2007Jun}; a proof for this is given in Appendix~\ref{app:late_time_bounds}.

The asymptotic result of Theorem~\ref{cor:product_decay} has an empirical extension: in both Gaussian and interacting settings (see Secs.~\ref{sec:single_body}-\ref{sec:many_body}), we find that $I(A{:}B)$ and $\mathcal{N}$ share peak, trough, and decay timescales throughout the full evolution, not only near the steady state. This tracking is not implied by Proposition~\ref{thm:factorization_joint} or Theorem~\ref{cor:product_decay}, which are asymptotic statements, but is instead established numerically in the following Sections. Appendix~\ref{app:quasiparticle} provides physical intuition for this tracking using a quasiparticle counting argument: both quantities receive contributions from the same entangled pairs shared across the cut and therefore rise, peak, and decay in tandem~\cite{AlbaCalabrese2019EPL,Bertini2022NegativityMI}.

When the Lindbladian respects a strong symmetry $S$~\cite{buca2012,albert2014} (i.e., $[L_\mu,S]=0$ for all $\mu$ and $[H,S]=0$), the dynamics is confined to fixed-charge sectors. Then, the NESS does not factorize across the bipartition because symmetry charge conservation introduces classical correlations between $A$ and $B$ that persist in the absence of quantum entanglement. One therefore finds
\begin{equation}
	\mathcal{N}(t)\xrightarrow{t\to\infty}0
	\quad\text{but}\quad
	I(A{:}B;t)\xrightarrow{t\to\infty}I_\infty>0,
	\label{eq:symmetry_plateau}
\end{equation}
consistent with Proposition~\ref{thm:factorization_joint} since the limiting state is not a product state. Any global conservation law or kinetic constraint that prevents asymptotic factorization sustains a non-zero plateau in $I(A{:}B)$ after $\mathcal{N}\to 0$; in such cases, the tracking of $\mathcal{N}$ by $I(A{:}B)$ breaks down at late times. This is demonstrated for the $S^z_{\mathrm{tot}}$--conserving XXZ chain under local $z$-dephasing in Sec.~\ref{sec:many_body}.
\\

Throughout this paper, we consider open boundary conditions and take the bipartition such that $A$ constitutes the leftmost $L_A$ sites. We compute $S_{1/2}(A)$, $I(A{:}B)$, and $\mathcal{N}(A{:}B)$ at each time step and use the unitary diagonal $\mathcal{N}=S_{1/2}$ as the benchmark against which the open-system trajectory is compared. The decoherence time $\tau_{\mathrm{d}}^{\mathrm{g}}$ is read off from the first loss of monotonicity in $\Gamma(t)$, and the long-time behavior of $I(A{:}B)$ is interpreted through Proposition~\ref{thm:factorization_joint} and Theorem~\ref{cor:product_decay} with attention to conservation laws that sustain classical correlations after entanglement has vanished.

\begin{figure*}[htbp]
\centering
\includegraphics[width=0.95\textwidth]{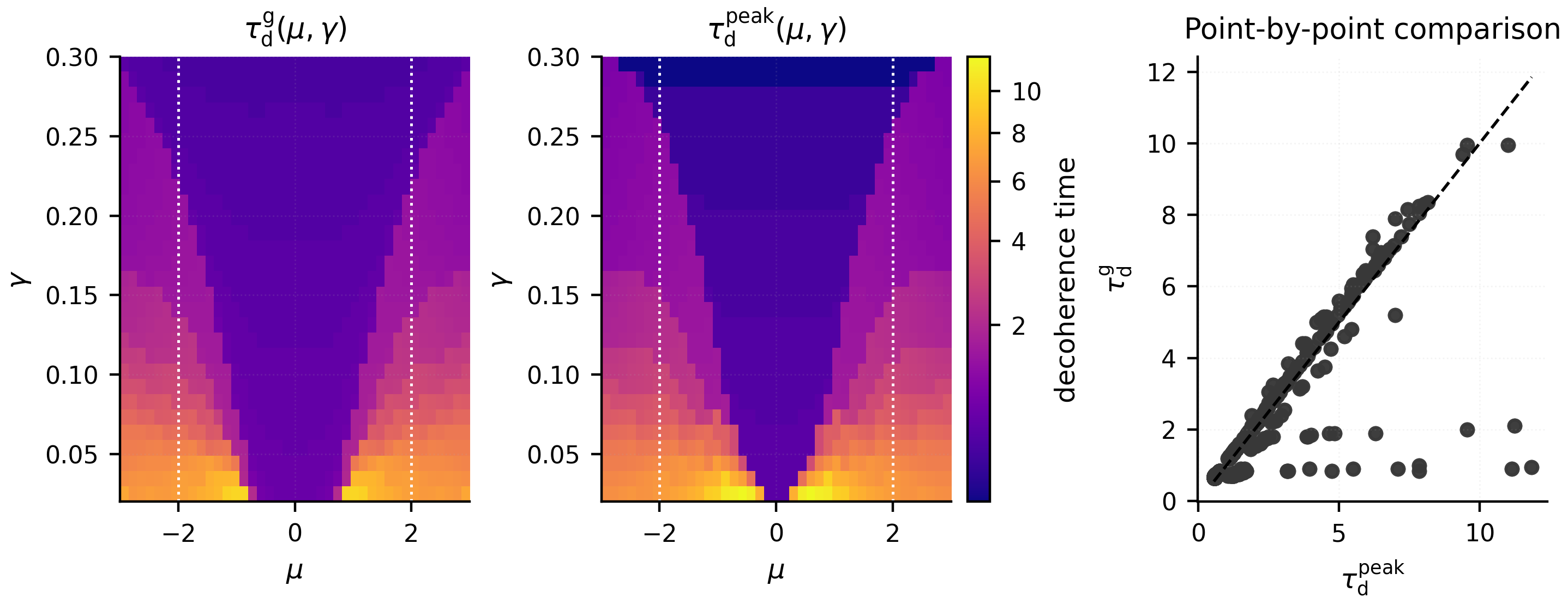}
\caption{Decoherence-time landscape for the balanced gain/loss Kitaev chain in the $(\mu,\gamma)$ plane. \emph{Left:} geometric decoherence time $\tau_{\mathrm{d}}^{\mathrm{g}}(\mu,\gamma)$. \emph{Center:} peak-based estimate $\tau_{\mathrm{d}}^{\mathrm{peak}}(\mu,\gamma)$. \emph{Right:} point-by-point comparison, with dashed diagonal $\tau_{\mathrm{d}}^{\mathrm{peak}}=\tau_{\mathrm{d}}^{\mathrm{g}}$. Vertical guide lines at $\mu=\pm 2J$ mark the closed-system topological phase boundaries. The topological phase $|\mu|<2J$ sustains longer $\tau_{\mathrm{d}}^{\mathrm{g}}$ than the trivial phase, and within the topological phase $\tau_{\mathrm{d}}^{\mathrm{g}}$ shows a local suppression at the chiral-symmetric point $\mu=0$ and decreases again toward $|\mu|=2J$. The comparison panel shows that $\tau_{\mathrm{d}}^{\mathrm{peak}}$ is correlated with $\tau_{\mathrm{d}}^{\mathrm{g}}$, while exhibiting qualitatively significant discrepancies in the weak-dissipation regime, as discussed in the main text. Parameters: $J=\Delta=1$, $L=96$, $L_A=\lfloor\sqrt{L}\rfloor=9$, N\'eel initial state, $\mu\in[-3,3]$ with step $0.15$, balanced gain/loss $\gamma_+=\gamma_-\equiv\gamma\in[0.02,0.30]$ with step $0.01$, time grid $t\in[0,12]$ with step $0.05$.}
\label{fig:gaussian_phase}
\end{figure*}

\section{Exactly Solvable Lindbladian Dynamics}
\label{sec:single_body}

We study the geometric decoherence time in the setting of exactly solvable Gaussian Lindbladian dynamics. For quadratic Hamiltonians and linear jump operators, the full many-body dynamics is encoded within the equal-time correlation matrix, and the diagnostics of Sec.~\ref{sec:theory}---the R\'enyi-$\tfrac{1}{2}$ entropy $S_{1/2}$, the fermionic Gaussian negativity $\mathcal{E}_F$, and the mutual information $I(A{:}B)$---can all be straightforwardly accessed. This setting offers two advantages over exact many-body evolution: it gives analytical control and enables thorough parameter scans over system size, subsystem size, chemical potential, and dissipation strength that are prohibitively expensive within the full Fock-space treatment.

\subsection{Model and Setup}
\label{subsec:model}

We consider the Kitaev chain~\cite{kitaev2001} on $L$ sites with open boundary conditions,
\begin{equation}
	\begin{aligned}
		H
		=
		-\mu\sum_{j=1}^{L}c_j^\dagger c_j
		& -J\sum_{j=1}^{L-1}\!\left(c_j^\dagger c_{j+1}+\mathrm{h.c.}\right) \\
		&+\Delta\sum_{j=1}^{L-1}\!\left(c_j c_{j+1}+\mathrm{h.c.}\right),
	\end{aligned}
	\label{eq:kitaev_ham_single_body}
\end{equation}
with $J>0$, $\Delta>0$, and we set $J=\Delta=1$ throughout. The topological phase occupies $|\mu|<2J$, with bulk gap closings at $|\mu|=2J$. Under open boundary conditions, the topological phase supports a pair of exponentially localized Majorana zero modes at the two ends of the chain, with localization length $\xi^{-1}=\ln|2J/\mu|$~\cite{kitaev2001}; the trivial phase $|\mu|>2J$ supports no such modes.
As we demonstrate below, these edge modes control the contrast in the geometric decoherence time between the topologically trivial and non-trivial phases.

In Nambu form, define the spinor $\Psi=(c_1,\ldots,c_L,c_1^\dagger,\ldots,c_L^\dagger)^T$, so that $H=\frac{1}{2}\Psi^\dagger\mathcal{H}_{\mathrm{BdG}}\Psi$ up to a constant, where $\mathcal{H}_{\mathrm{BdG}}$ is the $2L\times 2L$ Bogoliubov--de Gennes matrix. We couple the chain to a Markovian environment via single-site gain and loss jump operators
\begin{equation}
	L_{j,+}=\sqrt{\gamma_+}\,c_j^\dagger,
	\qquad
	L_{j,-}=\sqrt{\gamma_-}\,c_j,
\end{equation}
which preserve the Gaussian character of the state~\cite{Prosen2008ThirdQuantization,Prosen2010ThirdQuantization}. The full dynamics is therefore exactly encoded in the equal-time Nambu correlation matrix $C_{ab}(t)=\langle\Psi_a\Psi_b^\dagger\rangle_t$, which obeys the linear equation of motion
\begin{equation}
	\frac{dC}{dt}
	=
	-i[\mathcal{H}_{\mathrm{BdG}},C]
	-({\gamma_+}+{\gamma_-})C
	+
	\begin{pmatrix}
		\gamma_-\,\mathds{1}_L & 0\\
		0 & \gamma_+\,\mathds{1}_L
	\end{pmatrix}.
	\label{eq:correlation_matrix_ode}
\end{equation}
The three terms represent, respectively, unitary BdG rotation, uniform damping of all two-point correlators at rate $\gamma_++\gamma_-$, and a source that injects particles and holes to drive the system toward a mixed NESS with stationary correlation matrix (for either balanced gain/loss case or number-conserving case)
\begin{equation}
	C_\infty
	=
	\begin{pmatrix}
		\dfrac{\gamma_-}{\gamma_++\gamma_-}\,\mathds{1}_L & 0\\[6pt]
		0 & \dfrac{\gamma_+}{\gamma_++\gamma_-}\,\mathds{1}_L
	\end{pmatrix}.
	\label{eq:gaussian_steady_correlation}
\end{equation}
A derivation of Eq.~\eqref{eq:correlation_matrix_ode} is given in Appendix~\ref{app:derivation:eom}.

In the balanced case $\gamma_+=\gamma_-\equiv\gamma$, the steady state is the maximally mixed Gaussian state $C_\infty=\frac{1}{2}\mathds{1}_{2L}$, and Eq.~\eqref{eq:correlation_matrix_ode} admits the closed-form solution
\begin{equation}
	C(t)
	=
	e^{-2\gamma t}\,e^{-i\mathcal{H}_{\mathrm{BdG}}t}\,C(0)\,
	e^{+i\mathcal{H}_{\mathrm{BdG}}t}
	+
	\bigl(1-e^{-2\gamma t}\bigr)\frac{\mathds{1}_{2L}}{2}.
	\label{eq:balanced_closed_form}
\end{equation}
Each matrix element of the unitarily evolving initial correlation matrix is multiplied by the common damping factor $e^{-2\gamma t}$, while the $\mathds{1}_{2L}/2$ term interpolates the state toward the infinite-temperature Gaussian fixed point. Eq.~\eqref{eq:balanced_closed_form} gives analytical control: one diagonalizes $\mathcal{H}_{\mathrm{BdG}}$ once, evolves phase factors spectrally, and applies a single exponential envelope. For imbalanced gain and loss the closed-form simplification no longer applies and we integrate Eq.~\eqref{eq:correlation_matrix_ode} with a high-order ODE solver. Benchmarks against exact many-body Lindblad evolution at small system sizes, and explicit balanced-versus-imbalanced comparisons, are collected in Appendix~\ref{app:benchmarks}.

We initialize the chain in the N\'eel state $|\psi_0\rangle=|1010\cdots\rangle$, corresponding to a diagonal initial correlation matrix built from the occupied even sites; we have verified that the qualitative phenomenology is insensitive to this choice. At each time step we compute three observables for the bipartition $A|B$, where $A$ is the leftmost $L_A$ sites: the R\'enyi-$\tfrac{1}{2}$ entropy $S_{1/2}(A)$, the fermionic Gaussian negativity $\mathcal{E}_F(A{:}B)$, and the von Neumann mutual information $I(A{:}B)$. The geometric decoherence time $\tau_{\mathrm{d}}^{\mathrm{g}}$ is identified from the parametric trajectory $\Gamma(t)=(S_{1/2}(t),\mathcal{E}_F(t))$ via Definition~\ref{def:geom_dec}.

\subsection{Results}
\label{subsec:results}

\paragraph*{Phase diagram.}
Figure~\ref{fig:gaussian_phase} shows the geometric decoherence time $\tau_{\mathrm{d}}^{\mathrm{g}}(\mu,\gamma)$ together with the peak-based estimate (see Eq~\eqref{eq:peak_dec})
over the $(\mu,\gamma)$ plane for balanced gain and loss, using the N\'eel initial state throughout.

The main result is that the topological phase sustains longer $\tau_{\mathrm{d}}^{\mathrm{g}}$ than the trivial phase at every dissipation strength studied. Within the trivial phase, $\tau_{\mathrm{d}}^{\mathrm{g}}$ is short and weakly dependent on $\mu$. Within the topological phase, the behaviour is non-monotonic in $\mu$: $\tau_{\mathrm{d}}^{\mathrm{g}}$ reaches a local minimum at the chiral-symmetric point $\mu=0$ (where $J=\Delta$), rises to a local maximum on either side, and decreases again as $|\mu|\to 2J$ toward the phase boundary.


The comparison panel rightmost in Fig.~\ref{fig:gaussian_phase} plots $\tau_{\mathrm{d}}^{\mathrm{peak}}$ against $\tau_{\mathrm{d}}^{\mathrm{g}}$ point by point. Many parameter values fall close to the diagonal, confirming that the negativity peak is a reasonable operational proxy over much of parameter space. However, the most important outliers occur at weak dissipation within the topological phase, where $\tau_{\mathrm{d}}^{\mathrm{peak}}$ can substantially overshoot $\tau_{\mathrm{d}}^{\mathrm{g}}$ because the negativity peak lags the geometric onset. By contrast, at fixed moderate dissipation varying $\mu$ produces only a slight undershoot with little qualitative change between representative cuts; we relegate a discussion of this secondary effect to Appendix~\ref{app:gaussian_fixed_gamma_mu}. This confirms that the two definitions are logically and quantitatively distinct, and that $\tau_{\mathrm{d}}^{\mathrm{g}}$ is a sharper diagnostic.

\begin{figure*}[htbp]
	\centering
	\includegraphics[width=\textwidth]{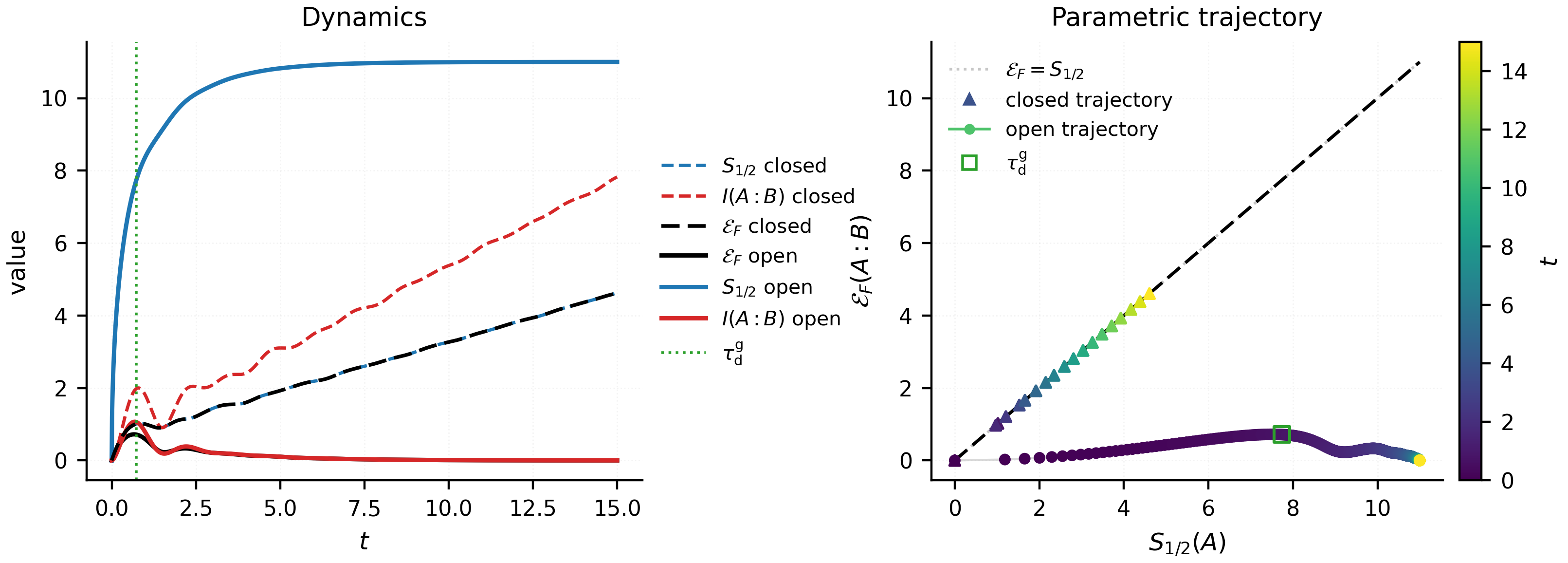}
\caption{Geometric decoherence time in the topological phase of the open Kitaev chain ($\mu=0.5$, $J=\Delta=1$, $L=128$, $L_A=11$, balanced gain/loss $\gamma=0.15$, N\'eel initial state). Left: time evolution of $\mathcal{E}_F$ (fermionic Gaussian negativity), $S_{1/2}(A)$, and $I(A{:}B)$ under unitary dynamics (dashed) and Lindbladian dynamics (solid); the vertical marker indicates $\tau_{\mathrm{d}}^{\mathrm{g}}$. Right: parametric trajectory $\Gamma(t)=(S_{1/2},\mathcal{E}_F)$ colored by time; triangles denote unitary dynamics, circles denote Lindbladian dynamics. Under unitary dynamics, the trajectory is pinned to the diagonal; under Lindbladian dynamics it bends away at $\tau_{\mathrm{d}}^{\mathrm{g}}$ (square marker). At late times $\mathcal{E}_F$ and $I(A{:}B)$ both vanish, signaling asymptotic factorization across the cut (Proposition~\ref{thm:factorization_joint}).}
	\label{fig:gaussian_story}
\end{figure*}

\paragraph*{Representative dynamics.}
Figure~\ref{fig:gaussian_story} illustrates the breakdown of the monotone relation that defines $\tau_{\mathrm{d}}^{\mathrm{g}}$ for a representative point in the topological phase ($\mu=0.5$, $\gamma=0.15$, $L=128$, $L_A=11$). The left panel shows the time evolution of $\mathcal{E}_F$, $S_{1/2}(A)$, and $I(A{:}B)$ under unitary dynamics (dashed) and Lindbladian dynamics with balanced gain/loss (solid). The right panel shows the parametric trajectory $\Gamma(t)=(S_{1/2}(t),\mathcal{E}_F(t))$, with unitary dynamics marked by triangles and Lindbladian dynamics marked by circles, colored by time.

Under unitary dynamics from the N\'eel initial state, $\Gamma(t)$ lies on the diagonal $\mathcal{E}_F=S_{1/2}$ at all times, consistent with the pure-state identity Eq.~\eqref{eq:pure_identity}. Under balanced gain/loss, the early-time trajectory follows the same monotonic growth: both $S_{1/2}$ and $\mathcal{E}_F$ rise together. At $t=\tau_{\mathrm{d}}^{\mathrm{g}}$ (vertical marker in the left panel; square marker in the right panel), the monotone relation first fails: $S_{1/2}$ continues to grow due to ongoing entropy production, while $\mathcal{E}_F$ begins to decrease, signalling the geometric onset of decoherence~\ref{def:geom_dec}. At late times, $\mathcal{E}_F\to 0$ and $I(A{:}B)\to 0$ together at the same rate, consistent with Proposition~\ref{thm:factorization_joint} and Theorem~\ref{cor:product_decay}, while $S_{1/2}$ saturates to a volume-law scale set by the boundary block.

\paragraph*{Mutual information tracking.}
Figure~\ref{fig:gaussian_tracking} shows $I(A{:}B;t)$ and $\mathcal{E}_F(t)$ on the same axes for the same representative parameters. The two quantities differ in magnitude, but share the same early-time growth, the same peak time to within the time-grid resolution ($\Delta t=0.01$), and the same asymptotic decay rate---exponential fits over the second half of the time window confirm $\lambda_{\mathcal{E}_F}\approx\lambda_{I}\approx 0.50$---consistent with Theorem~\ref{cor:product_decay}. This agreement in peak time, intermediate-time features, and decay envelope is the empirical content of the tracking claim of Sec.~\ref{subsec:MI_asymptotics}: in the absence of a strong symmetry that prevents asymptotic factorization, $I(A{:}B)$ reproduces the characteristic timescales of $\mathcal{E}_F$ throughout the nonequilibrium evolution, not only near the steady state. The robustness of this agreement under changes of bipartition and time-step resolution is verified in Appendix~\ref{app:benchmarks}, where we repeat the comparison at $L_A=64$ (half-chain cut) and at $\Delta t=0.005$, finding the same peak and decay times to within numerical precision.

\begin{figure}[htbp]
	\centering
	\includegraphics[width=0.82\columnwidth]{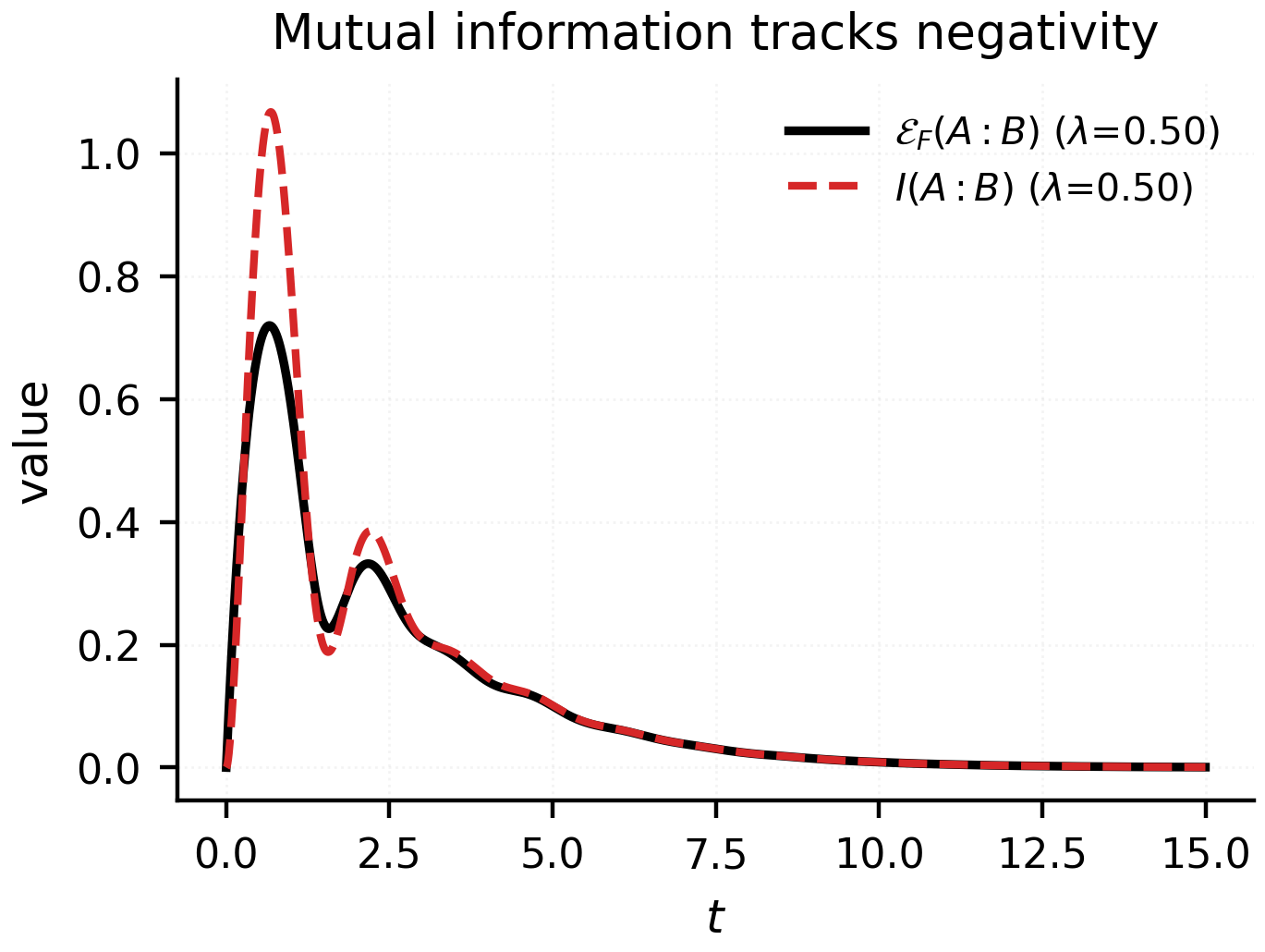}
\caption{Mutual information $I(A{:}B)$ and fermionic Gaussian negativity $\mathcal{E}_F$ as functions of time for the same parameters as Fig.~\ref{fig:gaussian_story}. The two quantities share the same peak and trough times as well as the same asymptotic decay rate, consistent with Theorem~\ref{cor:product_decay} (see Sec.~\ref{subsec:MI_asymptotics}). Exponential fits over the second half of the time window yield $\lambda_{\mathcal{E}_F} \approx \lambda_{I} \approx 0.50$, confirming a shared decay exponent to two significant figures.}
	\label{fig:gaussian_tracking}
\end{figure}

\subsection{Analysis}
\label{subsec:analysis}

\paragraph*{Slower Decoherence in the Topological Phase.}
The contrast between the topologically trivial and non-trivial phases shown in Fig.~\ref{fig:gaussian_phase} can be understood as follows. The linear super-operator governing the homogeneous part of Eq.~\eqref{eq:correlation_matrix_ode} has eigenvalues $\lambda_{mn}=-2\gamma-i(E_m-E_n)$, where $E_1,\ldots,E_{2L}$ are the eigenvalues of $\mathcal{H}_{\mathrm{BdG}}$. The Liouvillian spectral gap in the covariance sector is therefore $\Delta_{\mathrm{cov}}=2\gamma$, independent of $\mu$, $J$, and $\Delta$ (Appendix~\ref{app:derivation:gap}). The gap is identical in the topological phase, the trivial phase, and at the critical point $|\mu|=2J$: it cannot be the source of the different behaviour between the two phases. Any $\mu$-dependent variation of $\tau_{\mathrm{d}}^{\mathrm{g}}$ must therefore originate entirely from the \emph{unitary} part of the dynamics encoded in $e^{-i\mathcal{H}_{\mathrm{BdG}}t}$.

In the topological phase $|\mu|<2J$, the chain supports a pair of exponentially localized Majorana zero modes at its ends, with localization length $\xi^{-1}=\ln|2J/\mu|$~\cite{kitaev2001}. These modes generate nonlocal, spatially coherent correlations across the boundary bipartition that are largely insensitive to the on-site gain/loss perturbation: the jump operators $L_{j,\pm}$ act site by site and couple only weakly to modes that are exponentially small at any individual site. In the trivial phase, no such modes exist; correlations are bulk-dominated and decay on the scale of the correlation length, making them substantially more susceptible to local dissipation. This explains why the topological phase sustains larger $\tau_{\mathrm{d}}^{\mathrm{g}}$ at every value of $\gamma$.
 
The non-monotonic variation within the topological phase---a local minimum at $\mu=0$ flanked by maxima---can be attributed to the chiral symmetry at $\mu=0$ (with $J=\Delta$), which enforces an exact degeneracy between the two zero-mode branches. This degeneracy opens an additional channel for decoherence: the environment can couple the degenerate zero-mode subspace and suppress bipartite entanglement more efficiently than at generic $\mu$, shortening $\tau_{\mathrm{d}}^{\mathrm{g}}$. Moving away from $\mu=0$ lifts this degeneracy, reducing the susceptibility to dissipation and lengthening $\tau_{\mathrm{d}}^{\mathrm{g}}$. As $|\mu|\to 2J$ from within the topological phase, the bulk gap closes, the zero modes merge into the bulk continuum, and $\tau_{\mathrm{d}}^{\mathrm{g}}$ decreases again before dropping sharply across the phase boundary.

A leading-order perturbative expression for $\tau_{\mathrm{d}}^{\mathrm{g}}$ at small $\gamma$ follows from expanding the exact solution Eq.~\eqref{eq:balanced_closed_form} around the first local maximum $t_*$ of the unitary entanglement $s(t) = \mathcal{E}_F(t;\gamma=0) = S_{1/2}(t;\gamma=0)$:
\begin{equation}
	\tau_{\mathrm{d}}^{\mathrm{g}}
	= t_*
	- \gamma\,\frac{n_1'(t_*)}{s''(t_*)}
	+ O(\gamma^2),
	\label{eq:tau_perturbative}
\end{equation}
where $s''(t_*)<0$ is the curvature of the pure-state entanglement maximum, and $n_1(t) := \left.\frac{d}{d\gamma}\mathcal{E}_F(t;\gamma)\right|_{\gamma=0}$ is the linear-response of $\mathcal{E}_F$ to dissipation at fixed time. Accordingly, $n_1'(t_*)$ is the first-order response of $\mathcal{E}_F$ to $\gamma$ differentiated with respect to time at $t_*$. Since the balanced dissipator contributes the strictly decreasing envelope  $e^{-2\gamma t}$, the first-order suppression $n_1(t)$ grows in magnitude  monotonically with $t$ up to $t_*$, so $n_1'(t_*)<0$ (as verified  numerically); combined with $s''(t_*)<0$, the ratio $n_1'(t_*)/s''(t_*)>0$ and the shift is negative, confirming $\tau_{\mathrm{d}}^{\mathrm{g}}<t_*$ (as also verified numerically). Both $s''$ and $n_1'$ are purely unitary quantities determined by the $\gamma=0$ trajectory; $\gamma$ enters only as the prefactor of the shift. Here $t_*$ denotes the \emph{first} local maximum of the unitary $\mathcal{E}_F(t)$ before the onset of the first dip, not necessarily the global maximum; when the unitary trajectory has an early local maximum well before the global peak, it is this first local maximum that controls $\tau_{\mathrm{d}}^{\mathrm{g}}$ and to which the expansion applies. The derivation and consistency conditions are given in Appendix~\ref{app:perturbative}. The $\mu$-dependent sensitivity is encoded in the ratio $n_1'(t_*)/s''(t_*)$, whose variation with $\mu$ underlies the phase-dependent contrast in Fig.~\ref{fig:gaussian_phase}.

\begin{figure}[t]
	\centering
	\begin{subfigure}[t]{0.48\textwidth}
		\centering
		\includegraphics[width=\linewidth]{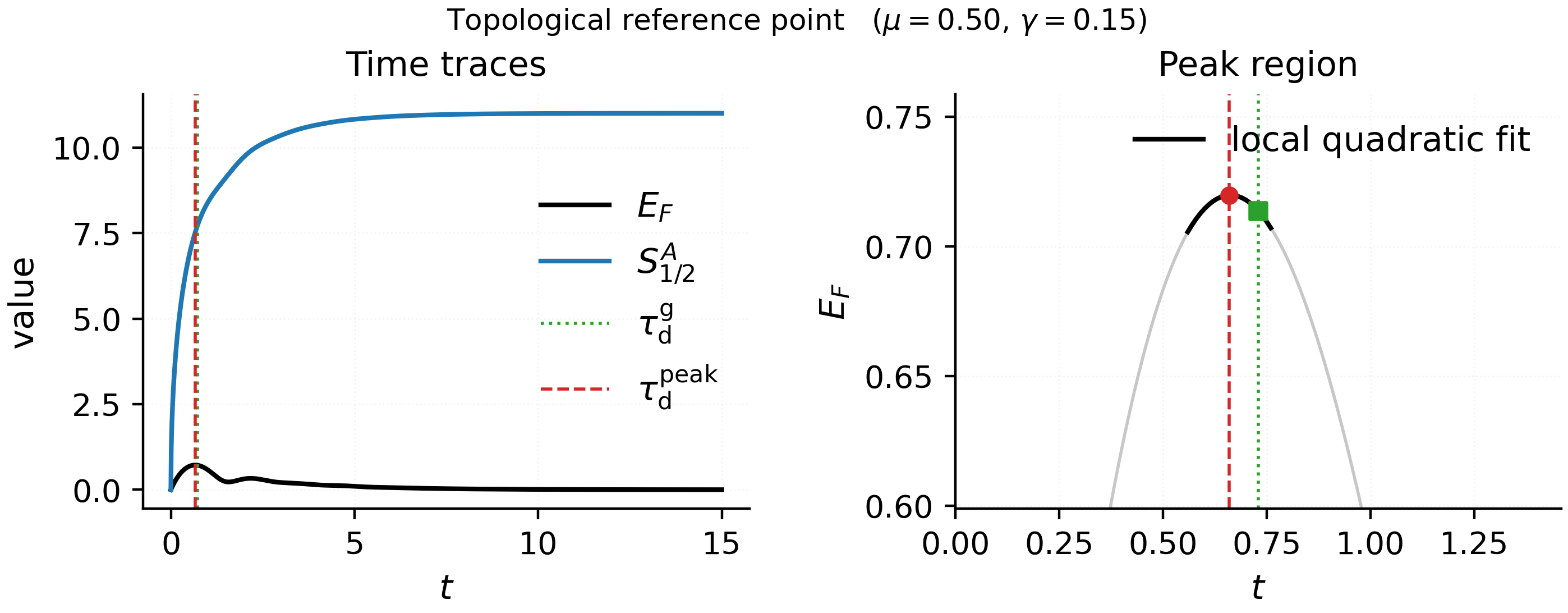}
		\caption{}
	\end{subfigure}\hfill
	\begin{subfigure}[t]{0.48\textwidth}
		\centering
		\includegraphics[width=\linewidth]{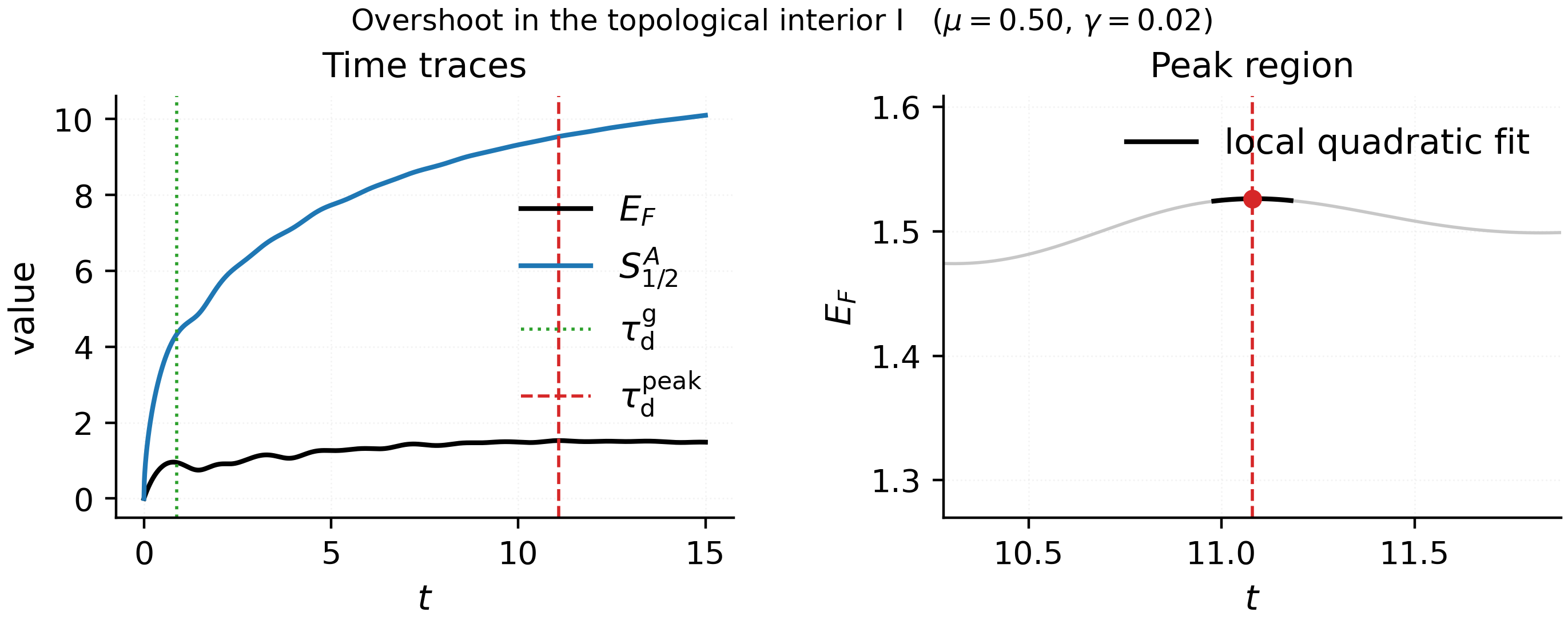}
		\caption{}
	\end{subfigure}
	\caption{Fixed-$\mu$ comparison within the topological phase. Left: topological reference point, $\mu=0.50$ and $\gamma=0.15$. Right: weaker dissipation, $\mu=0.50$ and $\gamma=0.02$. The weak-dissipation panel shows the clear overshoot of $\tau_{\mathrm{d}}^{\mathrm{peak}}$ relative to $\tau_{\mathrm{d}}^{\mathrm{g}}$ where the peak is attained at much later time and is comparatively flat. In both panels $\tau_{\mathrm{d}}^{\mathrm{g}}$ is extracted from Definition~\ref{def:geom_dec} via the entropy--negativity trajectory, which is not shown here.}
	\label{fig:gaussian_fixed_mu_gamma_compare}
\end{figure}

\begin{figure*}[htbp]
	\centering
	\begin{subfigure}[t]{0.48\textwidth}
		\centering
		\includegraphics[width=\linewidth]{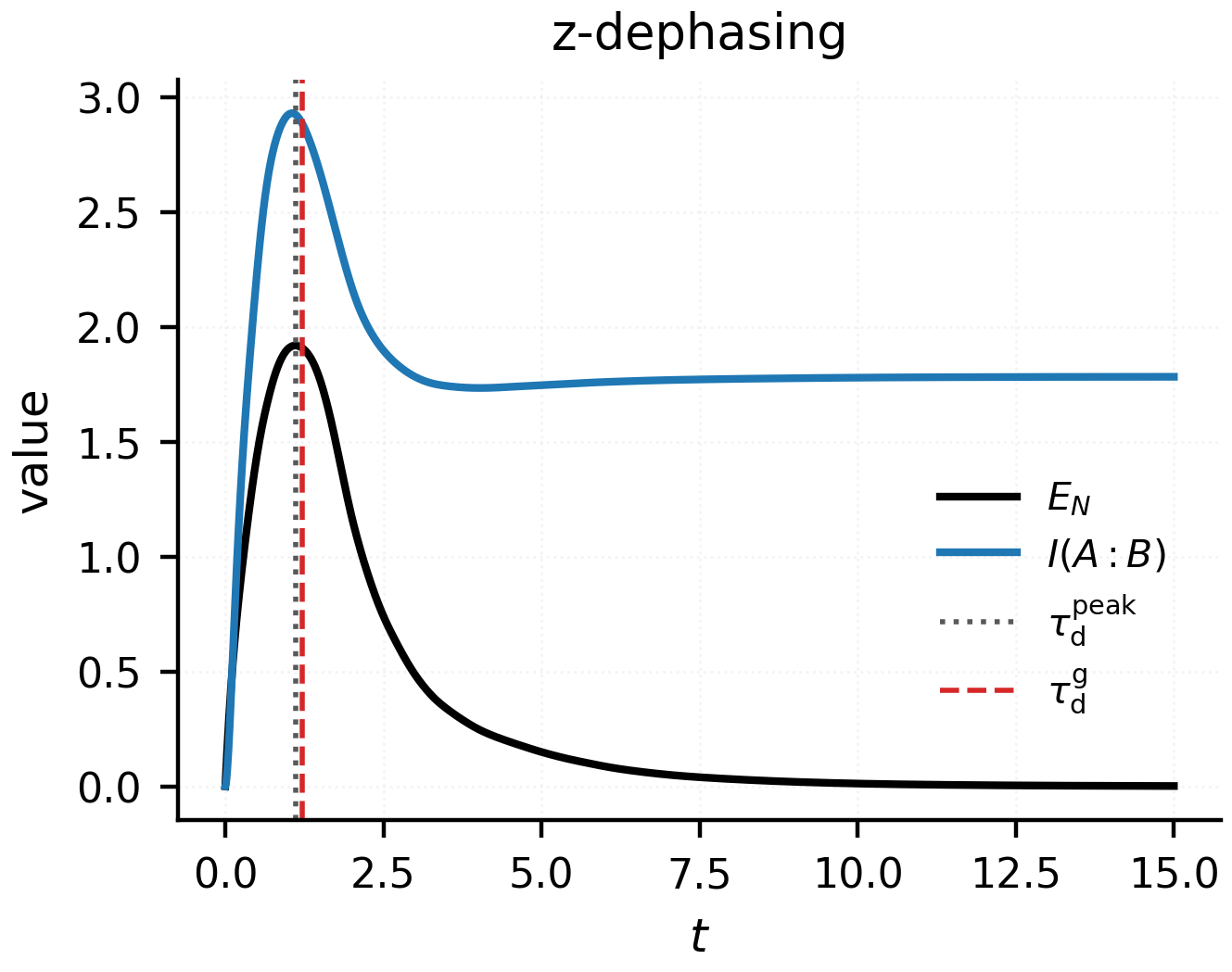}
		\caption{}
	\end{subfigure}\hfill
	\begin{subfigure}[t]{0.48\textwidth}
		\centering
		\includegraphics[width=\linewidth]{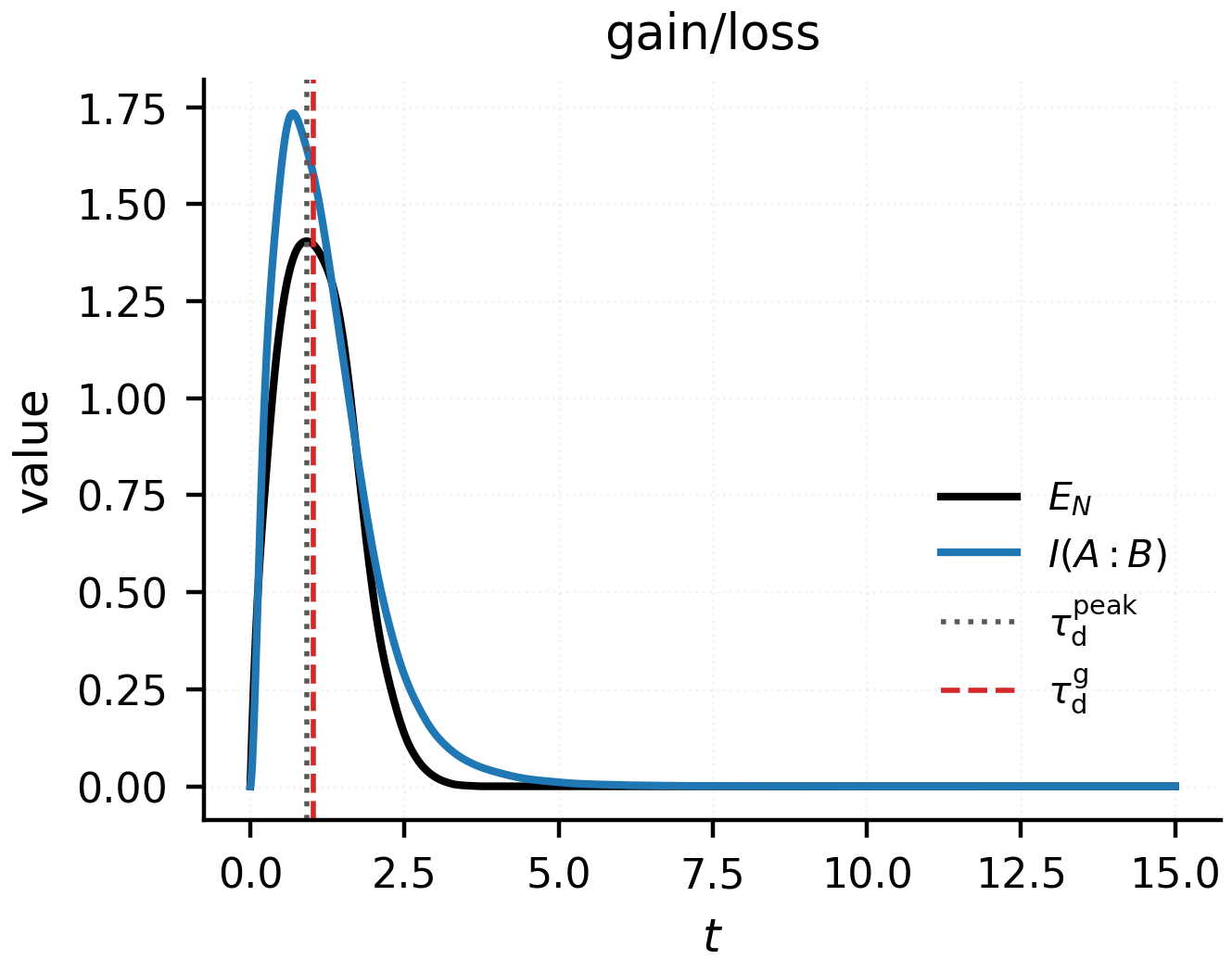}
		\caption{}
	\end{subfigure}
\caption{Time evolution of exact logarithmic negativity $E_N$ and von Neumann mutual information $I(A{:}B)$ for the interacting XXZ chain under two dissipative protocols. (a)~Under local $z$-dephasing (which preserves a strong U(1) symmetry), $E_N\to 0$ while $I(A{:}B)$ saturates to a nonzero plateau, signaling a separable but non-product NESS. The mutual-information tracking implied by Proposition~\ref{thm:factorization_joint} and Theorem~\ref{cor:product_decay} does not apply because the NESS does not factorize. (b)~Under local gain/loss, which removes the $S^z_{\mathrm{tot}}$ conservation, both $E_N$ and $I(A{:}B)$ decay to zero with aligned peak and decay times, consistent with asymptotic factorization (Proposition~\ref{thm:factorization_joint} and Theorem~\ref{cor:product_decay}). Parameters: open XXZ chain, $L=10$, half-chain bipartition $L_A=5$, $J=1$, $J_z=0.55$, N\'eel initial state, $t\in[0,15]$ with step $0.01$; panel~(a) uses $L_j=\sqrt{\gamma_z}\,\sigma_j^z$ with $\gamma_z=0.15$, panel~(b) uses $L_j^{(\pm)}=\sqrt{\gamma_\pm}\,\sigma_j^\pm$ with $\gamma_+=\gamma_-=0.15$.}
	\label{fig:many_body_tracking}
\end{figure*}

\paragraph*{Why $\tau_{\mathrm{d}}^{\mathrm{peak}}$ can fail.}
The main failure mode is a substantial overshoot at weak dissipation in the topological phase. This is seen by comparing the two fixed-$\mu$ cuts in Fig.~\ref{fig:gaussian_fixed_mu_gamma_compare} at $(\mu,\gamma)=(0.50,0.15)$ and $(0.50,0.02)$: lowering $\gamma$ leaves the geometric onset identifiable through Definition~\ref{def:geom_dec}, but the negativity curve develops a delayed turnover so that $\tau_{\mathrm{d}}^{\mathrm{peak}}$ occurs well after $\tau_{\mathrm{d}}^{\mathrm{g}}$. By contrast, fixing $\gamma$ and varying $\mu$ produces only a slight undershoot, with little qualitative change between the two cuts; we show this secondary effect in Appendix Fig.~\ref{fig:appendix_gaussian_fixed_gamma_mu}. Both effects reflect the same general point: $\tau_{\mathrm{d}}^{\mathrm{peak}}$ is defined by a later extremum of a single time trace, whereas $\tau_{\mathrm{d}}^{\mathrm{g}}$ is extracted from the geometry of the entropy--negativity trajectory $\Gamma(t)$.

\begin{figure*}[htbp]
	\centering
	\begin{subfigure}[t]{0.48\textwidth}
		\centering
		\includegraphics[width=\linewidth]{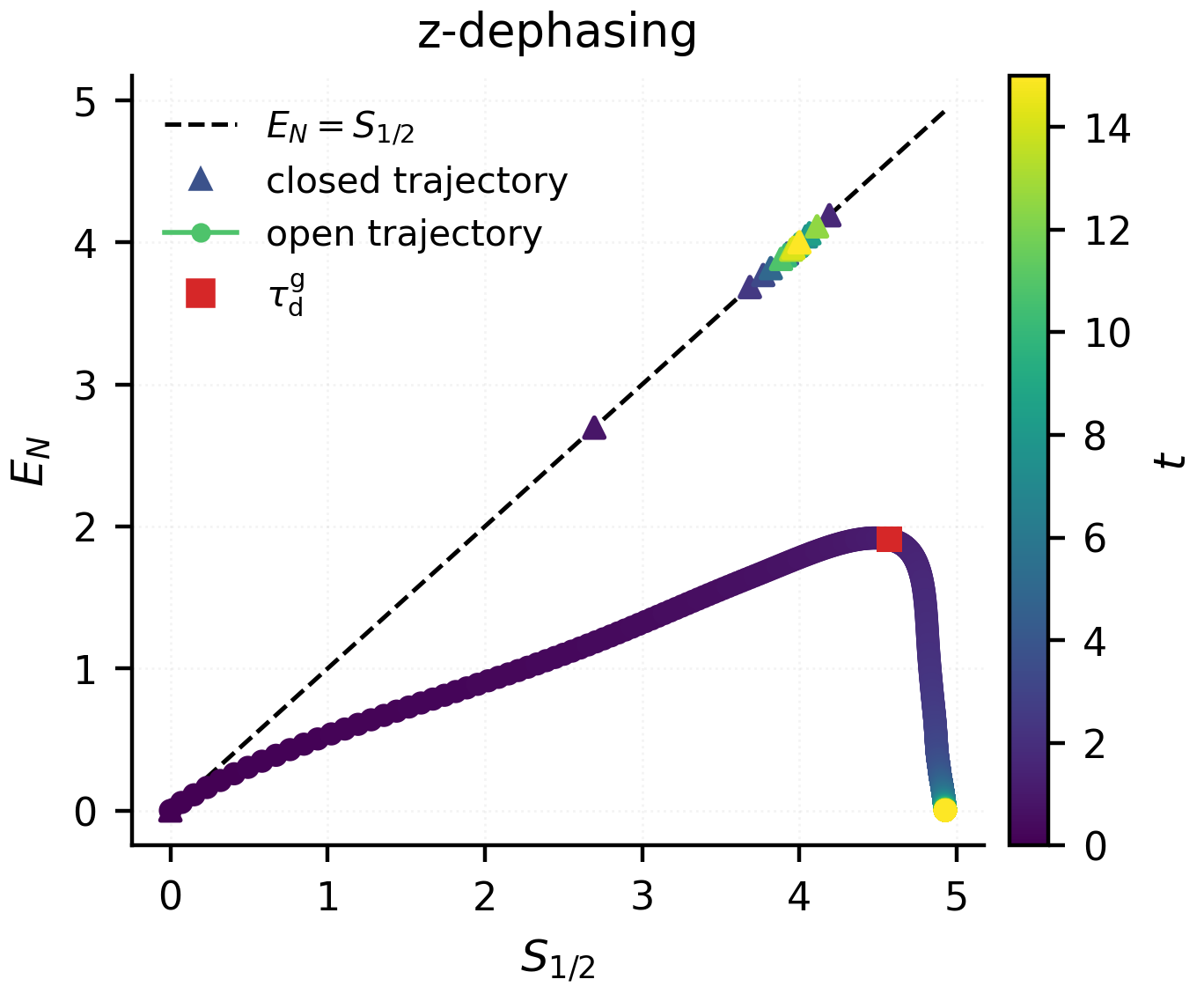}
		\caption{}
	\end{subfigure}\hfill
	\begin{subfigure}[t]{0.48\textwidth}
		\centering
		\includegraphics[width=\linewidth]{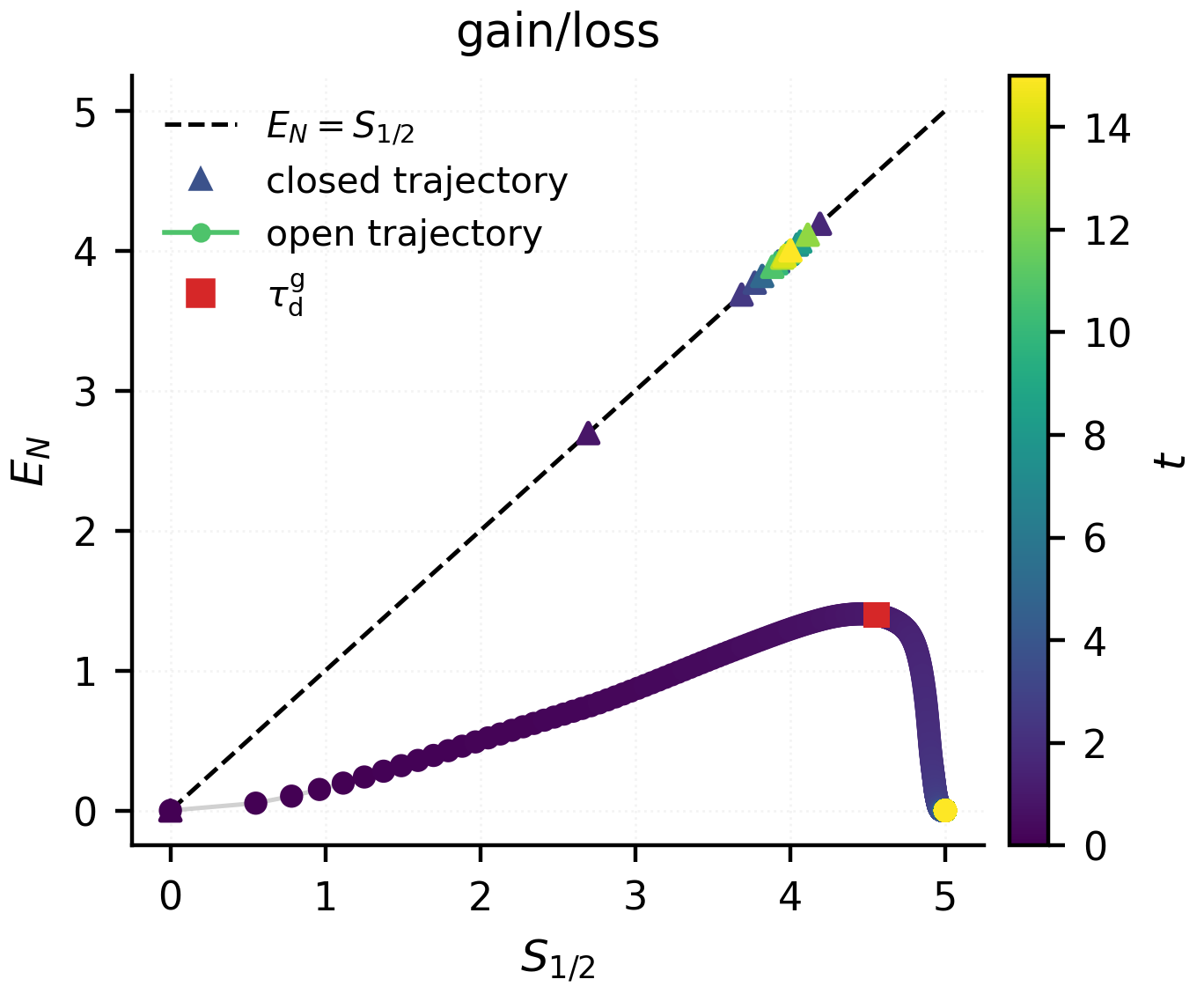}
		\caption{}
	\end{subfigure}
\caption{Entropy--negativity trajectories $\Gamma(t)=(S_{1/2}(t),E_N(t))$ for the two protocols of Fig.~\ref{fig:many_body_tracking}, colored by time. Both trajectories depart from the pure-state diagonal at $\tau_{\mathrm{d}}^{\mathrm{g}}$ (square marker), the earliest time at which $E_N$ decreases while $S_{1/2}$ continues to grow (Definition~\ref{def:geom_dec}). (a)~Under $z$-dephasing the trajectory respects conservation-law-protected classical correlations in the NESS. (b)~Under gain/loss the trajectory leads to a product NESS. Parameters as in Fig.~\ref{fig:many_body_tracking}.}
	\label{fig:many_body_parametric}
\end{figure*}

\paragraph*{Finite-size and bipartition analysis.}
Appendix~\ref{app:finite-size} collects two complementary scaling analyses: $\tau_{\mathrm{d}}^{\mathrm{g}}$ and the peak negativity as functions of $L_A$ at fixed $L=128$, and as functions of $L$ at fixed $L_A=16$. Both quantities remain flat across the ranges studied ($L_A\in\{3,5,\ldots,19\}$, $L\in\{64,96,\ldots,256\}$), confirming that the geometric decoherence scale has converged with respect to both system size and bipartition choice. The robustness of the mutual information tracking under changes of bipartition and time-step resolution is verified in Appendix~\ref{app:benchmarks}. These verify that the results in Fig.~\ref{fig:gaussian_phase} reflect the thermodynamic-limit dynamics and are not finite-size artifacts.

\section{Interacting Lindbladian Dynamics}
\label{sec:many_body}

We now test the framework of Sec.~\ref{sec:theory} beyond the Gaussian setting, using exact many-body Lindblad evolution of the interacting XXZ spin chain. The central question is whether the geometric decoherence criterion and the mutual-information tracking established analytically in Sec.~\ref{sec:theory} extend to the interacting regime, and how a strong symmetry alters the long-time behavior.

\paragraph*{Model.}
We consider the spin-$\tfrac12$ XXZ chain on $L$ sites with open boundary conditions,
\begin{equation}
	H_{\mathrm{XXZ}}
	=
	J\sum_{j=1}^{L-1}
	\!\left(S_j^x S_{j+1}^x + S_j^y S_{j+1}^y\right)
	+
	J_z \sum_{j=1}^{L-1} S_j^z S_{j+1}^z,
	\label{eq:xxz_ham}
\end{equation}
with $J=1$ and $J_z=0.55$ throughout. The Hamiltonian conserves total magnetization, $[H_{\mathrm{XXZ}},S^z_{\mathrm{tot}}]=0$, with $S^z_{\mathrm{tot}}=\sum_j S_j^z$.

We compare two dissipative protocols. For local $z$-dephasing, with jump operators $L_j=\sqrt{\gamma_z}\,\sigma_j^z$ that commute with $S^z_{\mathrm{tot}}$, the full Lindbladian posseses a strong U(1) symmetry~\cite{buca2012,albert2014}: both $[H_{\mathrm{XXZ}}, S^z_{\mathrm{tot}}]=0$ and $[L_j, S^z_{\mathrm{tot}}]=0$ hold for all $j$. The dynamics is consequently confined to fixed-magnetization sectors and the NESS cannot factorize across any bipartition that cuts bonds between $A$ and $B$, because total magnetization conservation enforces classical correlations between the two subsystems even in the absence of quantum entanglement. In contrast, local gain/loss, with jump operators $L_j^{(\pm)}=\sqrt{\gamma_\pm}\,\sigma_j^\pm$, does not commute with $S^z_{\mathrm{tot}}$ and magnetization is no longer conserved. In this case, only a weak U(1) symmetry is present and the typical route to asymptotic factorization across the bipartition is available.

\paragraph*{Results.}
Figure~\ref{fig:many_body_tracking} shows the time evolution of $E_N$ and $I(A{:}B)$ for both protocols at $L=10$, half-chain bipartition $L_A=5$, and the N\'eel initial state $|\psi_0\rangle=|\!\uparrow\downarrow\uparrow\downarrow\cdots\rangle$.

Under $z$-dephasing (Fig.~\ref{fig:many_body_tracking}(a)), $E_N$ decays to zero while $I(A{:}B)$ saturates to a nonzero plateau $I_\infty>0$. This is precisely the behavior anticipated by Eq.~\eqref{eq:symmetry_plateau}: the NESS does not factorize across the bipartition because total magnetization conservation forces finite classical correlations between $A$ and $B$. 
By Proposition~\ref{thm:factorization_joint}, the nonzero plateau signals that the NESS is separable but not a product state; Theorem 1 does not apply because its conclusion requires factorization.

Under gain/loss (Fig.~\ref{fig:many_body_tracking}(b)), the strong symmetry is absent. Both $E_N$ and $I(A{:}B)$ decay to zero at long times, and their peak times and decay timescales remain closely aligned throughout the nonequilibrium evolution. Proposition~\ref{thm:factorization_joint} and Theorem~\ref{cor:product_decay} both apply here: the NESS factorizes and both diagnostics share the asymptotic decay exponent. The robustness of these features under modest changes in system size is verified in Appendix~\ref{app:xxz-parametric-finite-size}.

Figure~\ref{fig:many_body_parametric} shows the entropy--negativity trajectories $\Gamma(t)=(S_{1/2}(t),E_N(t))$ for both protocols. In both cases the trajectory initially follows the pure-state diagonal but bends away at finite time when $E_N$ ceases to grow monotonically with $S_{1/2}$: this is the geometric onset of decoherence, $\tau_{\mathrm{d}}^{\mathrm{g}}$, in the sense of Definition~\ref{def:geom_dec}. The subsequent evolution differs sharply between the two protocols. Under $z$-dephasing the trajectory terminates at a point with $E_N=0$ but $S_{1/2}>0$, reflecting the separable but non-product NESS. Under gain/loss the trajectory approaches the origin, consistent with asymptotic factorization. The square marker in each panel identifies $\tau_{\mathrm{d}}^{\mathrm{g}}$ explicitly, confirming that the geometric criterion resolves the onset of decoherence in the interacting setting as well.

\section{Conclusion and Outlook}
\label{sec:conclusion}

We have introduced the geometric decoherence time $\tau_{\mathrm{d}}^{\mathrm{g}}$ as a dynamical scale for the onset of decoherence in open bipartite quantum systems, defined as the earliest moment the monotone relation between logarithmic negativity and R\'enyi-$\tfrac{1}{2}$ entropy (exactly equal across any bipartition for pure states~\cite{VidalWerner2002Negativity,Calabrese2012NegativityQFT,Calabrese2004EntanglementQFT}) fails. This definition signals the onset of entropy growth without accompanying entanglement growth, applies to Gaussian and interacting Lindbladian dynamics, and requires no spectral input or late-time fitting. A natural comparison scale is $\tau_{\mathrm{d}}^{\mathrm{peak}}$, the time at which the negativity attains its global maximum~\cite{Coser2014NegativityQuench,Sherman2016NegativityFiniteT,AlbaCarollo2022OpenNegativity}. Our results show that this prescription can substantially overshoot the geometric onset at weak dissipation within the topological phase of the Kitaev chain, confirming that these definitions are qualitatively distinct. 

We have also shown that the quantum mutual information $I(A{:}B) \to 0$ iff the NESS factorizes across the bipartition (Proposition~\ref{thm:factorization_joint}), and that exponential trace norm convergence to such a fixed point forces $I(A{:}B)$ and $E_N$ to share the same asymptotic decay rate (Theorem~\ref{cor:product_decay}). Beyond the asymptotic regime, we have numerically found that the two quantities track each other throughout the Lindbladian dynamics, except when a strong symmetry prevents factorization, in which case $E_N\to 0$ while $I(A{:}B)$ saturates to a finite plateau.

These results have been demonstrated via the Kitaev chain with balanced gain and loss, where a closed-form solution for the correlation matrix permits us to map $\tau_\mathrm{d}^\mathrm{g}$ across various chemical potentials and dissipation strengths. A key finding is that the topological phase sustains longer decoherence times than the trivial phase at every dissipation strength studied. This contrast can be traced to the existence of Majorana edge modes, which couple only weakly to the local jump operators. Within the topological phase, the chiral-symmetric point displays a local minimum for $\tau_{\mathrm{d}}^\mathrm{g}$. We have also verified that our results continue to hold for the interacting XXZ chain (for numerically accessible system-sizes). When a strong U(1) symmetry is preserved (local dephasing), the NESS is separable but not a product-state, with a finite $I(A{:}B)$ plateau; adding gain and loss breaks the strong U(1) symmetry to a weak U(1) symmetry, which restores asymptotic factorization.

A clear future direction is verifying the geometric criterion in higher-dimensional lattice models, in measurement-induced entanglement transitions~\cite{Skinner2019MIPT,GullansHuse2019Purification,Jian2021MIPTSYK}, and in random-state ensembles where negativity exhibits Page-curve-like behavior~\cite{Shapourian2021NegativitySpectrum}. Exploring these systems would test the model-independence claimed here and could expose new universality classes for the geometric onset of decoherence. In dissipative many-body engineering, $\tau_{\mathrm{d}}^{\mathrm{g}}$ can serve as a concrete design target for tailoring Lindblad generators in driven-dissipative spin and superconducting-qubit architectures~\cite{Aydogan2025StabilizingNESS,Yamanaka2023OpenXX,Zhang2024KitaevSuperconductingQubits,Shustin2026DissipativeTS}, with quadratic models admitting analytical control via correlation matrices and interacting systems addressed via large-scale tensor-network methods such as TDVP and MPO-based Lindblad evolution, which would also enable a systematic finite-size scaling analysis of $\tau_{\mathrm{d}}^{\mathrm{g}}$ with $L$. On the experimental side, the two key observables---logarithmic negativity (or its moments) and R\'enyi entropy or mutual information---are now accessible across several platforms: randomized-measurement and classical-shadow protocols~\cite{Elben2020MixedState,Zhou2020Negativity}, continuous-variable optical experiments~\cite{Jeng2019MBEDistillation}, and quantum-gas microscopes with single-site-resolved tomography~\cite{Bakr2009QuantumGasMicroscope,Gluza2020QuantumReadout}. Realistic routes for measuring the entropy–negativity trajectory now exist: engineered Lindbladian dynamics in quadratic and spin-chain models~\cite{Yamanaka2023OpenXX,Aydogan2025StabilizingNESS,Zhang2024KitaevSuperconductingQubits,Shustin2026DissipativeTS} and effective descriptions of inelastic loss in ultracold-atom settings~\cite{Braaten2017LindbladLoss} provide accessible platforms for testing whether the earliest loss of monotonicity furnishes an operational decoherence scale.

\section*{Acknowledgments}
R.J. acknowledges partial support, and S.H. acknowledges support, from the U.S. Department of Energy, Office of Science, Office of Advanced Scientific Computing Research, through the Exploratory Research for Extreme Scale Science (EXPRESS) program under Award No. DE-SC0026337. A.P. acknowledges support from the same program under Award No. DE-SC0026216.
	
\textit{Disclaimer.---} This report was prepared as an account of work sponsored by an agency of the United States Government. Neither the United States Government nor any agency thereof, nor any of their employees, makes any warranty, express or implied, or assumes any legal liability or 	responsibility for the accuracy, completeness, or usefulness of any information, apparatus, product, or process disclosed, or represents that its use would not infringe privately owned rights. Reference herein to any specific commercial product, process, or service by trade name, trademark, manufacturer, or otherwise does not necessarily constitute or imply its endorsement, recommendation, or favoring by the United States Government or any agency thereof. The views and opinions of authors expressed herein do not necessarily state or reflect those of the United States Government or any agency thereof.
	
\textit{Data availability.---} 	The data that support the findings of this article are available upon reasonable request from the authors.

	\bibliography{refs_gdt}
	
	\appendix
	
\clearpage

\section{The Pure-State Identity $E_N = S_{1/2}$}
\label{app:pure_state_proof}

The identity Eq.~\eqref{eq:pure_identity} is well known and follows from the Schmidt decomposition; it is implicit in Ref.~\cite{VidalWerner2002Negativity} and stated explicitly in the quantum field theory context in Refs.~\cite{Calabrese2012NegativityQFT,Calabrese2004EntanglementQFT}. We present a self-contained proof here for completeness.

\begin{proof}
Let $|\psi\rangle_{AB}=\sum_i\sqrt{p_i}\,|i_A\rangle\otimes|i_B\rangle$ be the Schmidt decomposition, with $p_i\ge 0$ and $\sum_i p_i=1$.
The reduced state is $\rho_A=\sum_i p_i|i_A\rangle\langle i_A|$, giving
	\begin{equation}
		S_{1/2}(\rho_A) = 2\log\!\left(\sum_i\sqrt{p_i}\right).
	\end{equation}
The partial transpose with respect to $B$ is
	\begin{equation}
		\rho_{AB}^{T_B}
		= \sum_{i,j}\sqrt{p_ip_j}\,
		|i_A\rangle\langle j_A|\otimes|j_B\rangle\langle i_B|.
	\end{equation}
The diagonal blocks $i=j$ yield eigenvectors $|i_A\rangle\otimes|i_B\rangle$ with eigenvalue $p_i$. For each off-diagonal pair $i\neq j$, the subspace $\mathrm{span}\{|i_A\rangle\otimes|j_B\rangle,\,|j_A\rangle\otimes|i_B\rangle\}$ is invariant with matrix representation $\bigl(\begin{smallmatrix}0&\sqrt{p_ip_j}\\\sqrt{p_ip_j}&0\end{smallmatrix}\bigr)$, whose eigenvalues are $\pm\sqrt{p_ip_j}$. The full spectrum of $\rho_{AB}^{T_B}$ is therefore $\{p_i\}_i\cup\{\pm\sqrt{p_ip_j}\}_{i<j}$, so
	\begin{equation}
		\left\|\rho_{AB}^{T_B}\right\|_1
		= \sum_ip_i+2\sum_{i<j}\sqrt{p_ip_j}
		= \left(\sum_i\sqrt{p_i}\right)^{\!2}.
	\end{equation}
Hence $E_N(\rho_{AB})=\log\|\rho_{AB}^{T_B}\|_1=S_{1/2}(\rho_A)$. The equality $S_{1/2}(\rho_A)=S_{1/2}(\rho_B)$ follows because $\rho_A$ and $\rho_B$ share the same nonzero spectrum $\{p_i\}$.
\end{proof}

The result holds for any bipartition and any pure state. In particular it holds throughout any unitary evolution from a pure initial state, irrespective of system size. In the fermionic Gaussian setting, the same identity holds with $E_N$ replaced by $\mathcal{E}_F$ when the state is Gaussian and pure, since Eq.~\eqref{eq:fermionic_negativity_spectrum} reduces to the Schmidt-spectrum expression in that limit~\cite{Shapourian2017FermionicNegativity,AlbaCarollo2022OpenNegativity}.

\section{Proof of Proposition~\ref{thm:factorization_joint}}
\label{app:factorization_proof}

We prove that, for a Lindblad evolution admitting a trace-norm limit $\rho(t)\to\rho_\infty$, the steady-state mutual information vanishes if and only if the steady state factorizes across the bipartition.

\begin{proof}
Recall from Eq.~\eqref{eq:MI_relative_entropy} that
	\begin{equation}
		I(A{:}B)_\sigma
		= D\!\left(\sigma\,\middle\|\,\sigma_A\otimes\sigma_B\right)
		\ge 0,
	\end{equation}
with equality if and only if $\sigma = \sigma_A\otimes\sigma_B$ (Klein's inequality~\cite{NielsenChuangBook}).
	
	\smallskip
	\noindent\textit{(i) implies (ii).}
If $\rho_\infty = \rho_{\infty,A}\otimes\rho_{\infty,B}$, then
	\begin{equation}
		I(A{:}B)_{\rho_\infty}
		= D\!\left(\rho_{\infty,A}\otimes\rho_{\infty,B}
		\,\middle\|\,\rho_{\infty,A}\otimes\rho_{\infty,B}\right)
		= 0.
	\end{equation}
Since $\rho(t)\to\rho_\infty$ in trace norm and mutual information is continuous in trace norm on any finite-dimensional Hilbert space~\cite{NielsenChuangBook}, it follows that $I(A{:}B;t)\to I(A{:}B)_{\rho_\infty}=0$.
	
	\smallskip
	\noindent\textit{(ii) implies (i).}
Suppose $I(A{:}B;t)\to 0$ and $\rho(t)\to\rho_\infty$ in trace norm. By continuity of $I$ in finite dimensions, $I(A{:}B;t)\to I(A{:}B)_{\rho_\infty}$. Since the left-hand side tends to zero, $I(A{:}B)_{\rho_\infty}=0$, which by Eq.~\eqref{eq:MI_iff_product} implies $\rho_\infty = \rho_{\infty,A}\otimes\rho_{\infty,B}$.
	
	\smallskip
	\noindent\textit{Vanishing of $E_N$.}
If $\rho_\infty$ is a product state, its partial transpose is
	\begin{equation}
		\left(\rho_{\infty,A}\otimes\rho_{\infty,B}\right)^{T_B}
		= \rho_{\infty,A}\otimes\rho_{\infty,B}^{\,T} \ge 0,
	\end{equation}
and for any positive semidefinite operator the trace norm equals the trace, so $\|\rho_\infty^{T_B}\|_1 = 1$ and $E_N(\rho_\infty)=0$. Since partial transposition is a bounded linear map on trace-class operators in finite dimensions, and $\rho(t)\to\rho_\infty$ in trace norm, one obtains $\|\rho(t)^{T_B}\|_1\to\|\rho_\infty^{T_B}\|_1=1$, hence $E_N(t)\to 0$.
\end{proof}

\section{Sufficient Condition: Unital Primitive Dynamics}
\label{app:unital_primitive}

Let $\{\Phi_t\}_{t\ge 0}$ denote the quantum dynamical semigroup generated by the GKLS equation~\eqref{eq:lindblad_general}, with $\rho(t)=\Phi_t[\rho(0)]$, acting on the Hilbert space $\mathcal{H}=\mathcal{H}_A\otimes\mathcal{H}_B$ with $d=d_Ad_B=\dim\mathcal{H}$. Recall that the semigroup is called \emph{unital} if $\Phi_t(\mathds{1})=\mathds{1}$ for all $t\ge 0$ (equivalently, the maximally mixed state is a fixed point) and \emph{primitive} if it has a unique full-rank stationary state $\rho_\infty$ such that $\Phi_t[\rho]\to\rho_\infty$ in trace norm for every initial state $\rho$~\cite{Wolf2012QChannels}.

\begin{remark}[Unital primitive dynamics as a sufficient condition]
If $\{\Phi_t\}$ is both unital and primitive, then its unique stationary state is the maximally mixed state $\rho_\infty=\mathds{1}/d$. Since $\mathds{1}=\mathds{1}_A\otimes\mathds{1}_B$, this factorizes as
	\begin{equation}
		\rho_\infty
		= \frac{\mathds{1}_A}{d_A}\otimes\frac{\mathds{1}_B}{d_B},
	\end{equation}
satisfying condition~(i) of Proposition~\ref{thm:factorization_joint}. Consequently,
	\begin{equation}
		I(A{:}B;t)\xrightarrow{t\to\infty}0,
		\qquad
		E_N(t)\xrightarrow{t\to\infty}0.
	\end{equation}
This provides the simplest concrete realization of the theorem: infinite temperature is sufficient but not necessary for factorization. The balanced gain/loss Kitaev chain studied in Sec.~\ref{sec:single_body} is an example: the balanced NESS is $C_\infty=\frac{1}{2}\mathds{1}_{2L}$, corresponding to $\rho_\infty\propto\mathds{1}$, so both $I$ and $\mathcal{E}_F$ vanish asymptotically regardless of initial state. Decoherence in the sense of Definition~\ref{def:geom_dec} occurs at finite time during the approach to this fixed point.
\end{remark}

\section{Quantitative late-time bounds near a product fixed point}
\label{app:late_time_bounds}

Here, we prove Theorem~\ref{cor:product_decay}. We work with a finite dimensional local Hilbert space and assume that the Lindblad evolution satisfies the trace-norm mixing bound
\begin{equation}
	\|\rho(t)-\rho_\infty\|_1 \le C e^{-\lambda t}
	\label{eq:product_mixing_bound}
\end{equation}
for constants $C,\lambda>0$ and all sufficiently large $t$, and that $\rho_\infty=\rho_{\infty,A}\otimes\rho_{\infty,B}$ is a product state. We define
\begin{equation}
	\begin{aligned}
		\rho_{\infty,A} &:= \operatorname{Tr}_B \rho_\infty, \\
		\rho_{\infty,B} &:= \operatorname{Tr}_A \rho_\infty, \\
		I(A{:}B;t) &:= I(A{:}B)_{\rho(t)},
	\end{aligned}
\end{equation}
and let
\begin{equation}
	\varepsilon(t) := \|\rho(t)-\rho_\infty\|_1,
	\qquad
	\eta(t) := \frac{\varepsilon(t)}{2}.
\end{equation}
If $d_A=1$ or $d_B=1$, both $I(A{:}B;t)$ and $E_N(t)$ vanish identically and there is nothing to prove. We therefore assume $d_A,d_B\ge 2$.

\paragraph*{Bound on $I(A{:}B;t)$.}
Because partial trace is CPTP, the trace norm is contractive, so
\begin{equation}
	\|\rho_A(t)-\rho_{\infty,A}\|_1 \le \varepsilon(t),
	\qquad
	\|\rho_B(t)-\rho_{\infty,B}\|_1 \le \varepsilon(t).
	\label{eq:appendix_reduced_trace_bound}
\end{equation}
Since $I(A{:}B)_\rho = S(\rho_A)+S(\rho_B)-S(\rho_{AB})$ and $I(A{:}B)_{\rho_\infty}=0$, the triangle inequality gives
\begin{equation}
	\begin{aligned}
		I(A{:}B;t) \le{} & |S(\rho_A(t))-S(\rho_{\infty,A})| \\
		&+ |S(\rho_B(t))-S(\rho_{\infty,B})| \\
		&+ |S(\rho(t))-S(\rho_\infty)|.
	\end{aligned}
	\label{eq:appendix_MI_triangle}
\end{equation}
The Audenaert--Fannes continuity bound~\cite{Fannes1973Dec,Audenaert2007Jun} states that for states on a finite-dimensional Hilbert space $\mathcal{K}$ with $\frac{1}{2}\|\sigma-\tau\|_1\le\eta\le\tfrac{1}{2}$,
\begin{equation}
	|S(\sigma)-S(\tau)|
	\le \eta\log(\dim\mathcal{K}-1)+h_2(\eta),
	\label{eq:appendix_Audenaert_Fannes}
\end{equation}
where $h_2(x):=-x\log x-(1-x)\log(1-x)$. Since $\eta(t)\to 0^+$ under Eq.~\eqref{eq:product_mixing_bound}, the condition $\eta(t)\le\tfrac{1}{2}$ holds for all sufficiently large $t$. Applying Eq.~\eqref{eq:appendix_Audenaert_Fannes} to the three terms in Eq.~\eqref{eq:appendix_MI_triangle} with $d=d_Ad_B$ and using $d_A,d_B\le d$ yields
\begin{equation}
	I(A{:}B;t)
	\le \frac{3}{2}\,\varepsilon(t)\log d
	+ 3\,h_2\!\left(\eta(t)\right).
	\label{eq:appendix_MI_final_bound}
\end{equation}
Since $h_2(x) = x\log(1/x)+O(x)$ as $x\to 0$,
\begin{equation}
	h_2\!\left(\eta(t)\right)
	= O\!\left(\eta(t)\log\tfrac{1}{\eta(t)}\right)
	= O\!\left(t\,e^{-\lambda t}\right).
\end{equation}
Together with $\varepsilon(t)=O(e^{-\lambda t})$, Eq.~\eqref{eq:appendix_MI_final_bound} gives
\begin{equation}
	I(A{:}B;t) = O\!\left(t\,e^{-\lambda t}\right).
\end{equation}

\paragraph*{Bound on $E_N(t)$.}
Since $\rho_\infty=\rho_{\infty,A}\otimes\rho_{\infty,B}$,
\begin{equation}
	\rho_\infty^{T_B}
	= \rho_{\infty,A}\otimes\rho_{\infty,B}^{\,T} \ge 0,
	\qquad
	\|\rho_\infty^{T_B}\|_1 = 1.
\end{equation}
Partial transposition is a bounded linear map on trace-class operators in finite dimensions, so there exists a constant $c_{T_B}>0$, depending only on the bipartition dimensions, such that
\begin{equation}
	\|X^{T_B}\|_1 \le c_{T_B}\|X\|_1
	\qquad \text{for all }X.
	\label{eq:appendix_partial_transpose_bound}
\end{equation}
By linearity of partial transposition and the triangle inequality,
\begin{equation}
	\|\rho(t)^{T_B}\|_1
	\le 1 + \|(\rho(t)-\rho_\infty)^{T_B}\|_1
	\le 1 + c_{T_B}\,\varepsilon(t).
\end{equation}
Since $ \log(1+x) \leq x/\ln 2 $,
\begin{equation}
	E_N(t)
	= \log\|\rho(t)^{T_B}\|_1
	\leq \frac{c_{T_B}}{\ln 2}\,\varepsilon(t)
	= O\!\left(e^{-\lambda t}\right).
\end{equation}
This completes the proof of Theorem~\ref{cor:product_decay}.

\section{Quasiparticle Counting Argument for Mutual Information Tracking}
\label{app:quasiparticle}

We present the quasiparticle counting argument that provides physical intuition for why $I(A{:}B)$ and $\mathcal{N}$ peak and decay on the same timescales in the absence of a strong symmetry.

Consider three disjoint regions $A_1$, $A_2$, and their joint complement $B=(A_1\cup A_2)^c$. A quasiparticle pair created at some time $t_0$ and propagating ballistically can fall into one of four classes:
\begin{enumerate}[label=(\roman*)]
	\item one quasiparticle in $A_1$, one in $A_2$;
	\item one quasiparticle in $A_1$, one in $B$;
	\item one quasiparticle in $A_2$, one in $B$;
	\item both quasiparticles in the same region.
\end{enumerate}
In the standard quasiparticle picture~\cite{CalabreseCardy2005,AlbaCalabrese2019EPL}, $S(A_1)$ counts pairs split across the $A_1|A_1^c$ boundary, namely classes (i) and (ii); $S(A_2)$ counts classes (i) and (iii); and $S(A_1\cup A_2)$ counts classes (ii) and (iii). Consequently,
\begin{equation}
	I(A_1{:}A_2)
	= S(A_1)+S(A_2)-S(A_1\cup A_2)
	= 2\times\text{(class (i))},
	\label{eq:MI_quasiparticle}
\end{equation}
where the contributions from classes (ii) and (iii) cancel exactly. Equation~\eqref{eq:MI_quasiparticle} shows that mutual information retains precisely the contribution from pairs shared between $A_1$ and $A_2$, the same pairs responsible for the bipartite entanglement and hence for the negativity. This is why both quantities peak when the density of shared pairs is maximal and decay together as dissipation destroys those pairs.

In the present setup $A_1\cup A_2=A$ and $B=A^c$, so the argument applies with $I(A_1{:}A_2)$ replaced by $I(A{:}B)$ when the bipartition is between $A$ and its complement. The quasiparticle argument is not a rigorous bound between $I$ and $\mathcal{N}$ --- it applies strictly in free or weakly interacting systems and does not account for multipartite entanglement in genuinely interacting systems --- but it provides the correct physical picture for why the two quantities track each other throughout the nonequilibrium dynamics. When a strong symmetry is present, the counting breaks down because pairs can carry conserved charge that introduces classical correlations between $A$ and $B$ even after all quantum entanglement has been destroyed, consistent with Eq.~\eqref{eq:symmetry_plateau}.

\section{Correlation-Matrix Dynamics, Liouvillian Gap, and Perturbative Decoherence Time}
\label{app:derivation}

This Appendix has three parts. Section~\ref{app:derivation:eom} derives the equation of motion Eq.~\eqref{eq:correlation_matrix_ode} from the Lindblad master equation and establishes the exact balanced-case solution Eq.~\eqref{eq:balanced_closed_form}. Section~\ref{app:derivation:gap} derives the covariance-sector Liouvillian gap used in the analysis of Sec.~\ref{sec:single_body}.
Section~\ref{app:perturbative} proves the perturbative decoherence-time formula Eq.~\eqref{eq:tau_perturbative}.

\subsection{Equation of motion for $C(t)$ and the balanced solution}
\label{app:derivation:eom}

We consider a chain of $L$ sites with fermionic creation and annihilation operators satisfying $\{c_i,c_j^\dagger\}=\delta_{ij}$ and $\{c_i,c_j\}=0$. Collecting all $2L$ operators into the Nambu spinor $\Psi=(c_1,\ldots,c_L,c_1^\dagger,\ldots,c_L^\dagger)^T$, the Hamiltonian takes the form $H=\tfrac{1}{2}\Psi^\dagger\mathcal{H}_{\mathrm{BdG}}\Psi$ up to a constant, where $\mathcal{H}_{\mathrm{BdG}}$ is the $2L\times 2L$ Hermitian Bogoliubov--de Gennes matrix. The dynamical variable is the equal-time Nambu correlation matrix $C_{ab}(t):=\langle\Psi_a\Psi_b^\dagger\rangle_t$.

The Lindblad master equation Eq.~\eqref{eq:lindblad_general} is driven by the single-site jump operators $L_{j,+}=\sqrt{\gamma_+}\,c_j^\dagger$ and $L_{j,-}=\sqrt{\gamma_-}\,c_j$, with $\gamma_\pm$ rates of dimension $[\text{time}]^{-1}$. Because both the Hamiltonian and the jump operators are linear in the fermionic operators, an initially Gaussian state remains Gaussian throughout the evolution~\cite{Prosen2008ThirdQuantization,Prosen2010ThirdQuantization}, and the entire dynamics is encoded in $C(t)$. Taking $d\langle\Psi_a\Psi_b^\dagger\rangle/dt$ via the master equation and evaluating the resulting commutators and anticommutators using the canonical anticommutation relations gives, after standard algebra,
\begin{equation}
	\frac{dC}{dt}
	=
	-i[\mathcal{H}_{\mathrm{BdG}},C]
	-(\gamma_++\gamma_-)C
	+
	\begin{pmatrix}
		\gamma_-\,\mathds{1}_L & 0 \\
		0 & \gamma_+\,\mathds{1}_L
	\end{pmatrix},
	\label{eq:cmde_app}
\end{equation}
which is Eq.~\eqref{eq:correlation_matrix_ode} of the main text. The commutator term rotates $C$ unitarily under the BdG Hamiltonian; the second term damps all two-point correlators uniformly at rate $\gamma_++\gamma_-$; the source matrix injects particles (at rate $\gamma_-$) and holes (at rate $\gamma_+$) to drive the system toward a mixed steady state. Setting $dC/dt=0$ and using $[\mathcal{H}_{\mathrm{BdG}},C_\infty]=0$ for a diagonal $C_\infty$ gives
\begin{equation}
	C_\infty
	=
	\begin{pmatrix}
		\dfrac{\gamma_-}{\gamma_++\gamma_-}\,\mathds{1}_L & 0 \\[6pt]
		0 & \dfrac{\gamma_+}{\gamma_++\gamma_-}\,\mathds{1}_L
	\end{pmatrix},
	\label{eq:gaussian_steady_app}
\end{equation}
valid whenever $\gamma_++\gamma_->0$.

In the balanced case $\gamma_+=\gamma_-\equiv\gamma$, Eq.~\eqref{eq:cmde_app} simplifies to
\begin{equation}
	\frac{dC}{dt}
	= -i[\mathcal{H}_{\mathrm{BdG}},C] - 2\gamma C
	+ \gamma\,\mathds{1}_{2L}.
\end{equation}
Introducing the traceless deviation $X(t):=C(t)-\tfrac{1}{2}\mathds{1}_{2L}$ and using $[\mathcal{H}_{\mathrm{BdG}},\mathds{1}_{2L}]=0$ reduces this to the homogeneous equation
\begin{equation}
	\frac{dX}{dt}
	= -i[\mathcal{H}_{\mathrm{BdG}},X] - 2\gamma X.
	\label{eq:X_ode}
\end{equation}
Diagonalizing $\mathcal{H}_{\mathrm{BdG}}$ with eigenvalues $E_1,\ldots,E_{2L}$, the matrix elements of $X$ decouple:
\begin{equation}
	\frac{d(X)_{mn}}{dt}
	= \bigl[-i(E_m-E_n)-2\gamma\bigr](X)_{mn},
\end{equation}
with solution $(X(t))_{mn}=e^{[-i(E_m-E_n)-2\gamma]t}(X(0))_{mn}$, or in matrix form $X(t)=e^{-2\gamma t}e^{-i\mathcal{H}_{\mathrm{BdG}}t} X(0)\,e^{+i\mathcal{H}_{\mathrm{BdG}}t}$. Returning to $C=X+\tfrac{1}{2}\mathds{1}_{2L}$ yields the exact balanced solution
\begin{equation}
	C(t)
	=
	e^{-2\gamma t}\,
	e^{-i\mathcal{H}_{\mathrm{BdG}}t}\,C(0)\,
	e^{+i\mathcal{H}_{\mathrm{BdG}}t}
	+
	\bigl(1-e^{-2\gamma t}\bigr)\frac{\mathds{1}_{2L}}{2},
\end{equation}
which is Eq.~\eqref{eq:balanced_closed_form}. As $t\to\infty$, $C(t)\to\tfrac{1}{2}\mathds{1}_{2L}$, confirming that the maximally mixed Gaussian state is the unique stationary state for any initial condition when $\gamma>0$.

\subsection{Covariance-sector Liouvillian gap}
\label{app:derivation:gap}

The homogeneous equation~\eqref{eq:X_ode} is governed by the linear super-operator
\begin{equation}
	\mathcal{L}_{\mathrm{cov}}(X)
	:= -i[\mathcal{H}_{\mathrm{BdG}},\,X] - 2\gamma X,
	\label{eq:Lcov_def}
\end{equation}
where $X=C(t)-\tfrac{1}{2}\mathds{1}_{2L}$ is the traceless deviation from the infinite-temperature fixed point. Its spectrum is readily obtained. Since $\mathcal{H}_{\mathrm{BdG}}$ is Hermitian, it admits an orthonormal eigenbasis $\{\phi_n\}$ with real eigenvalues $E_n$. The rank-one matrices $|\phi_m\rangle\langle\phi_n|$ form a basis for all $2L\times 2L$ matrices, and a straightforward calculation gives
\begin{equation}
	\begin{aligned}
	\mathcal{L}_{\mathrm{cov}}\bigl(|\phi_m\rangle\langle\phi_n|\bigr)
	=&
	\lambda_{mn}\,|\phi_m\rangle\langle\phi_n|, \\
	\lambda_{mn} =& -2\gamma - i(E_m-E_n).
	\end{aligned}
	\label{eq:Lcov_spectrum}
\end{equation}
These $(2L)^2$ eigenvalues exhaust the spectrum of $\mathcal{L}_{\mathrm{cov}}$. Because $E_m-E_n$ is real, the real part of every eigenvalue equals $-2\gamma$ exactly, so the covariance-sector Liouvillian gap satisfies
\begin{equation}
	\Delta_{\mathrm{cov}} = 2\gamma,
	\label{eq:gap_result}
\end{equation}
independent of $\mu$, $J$, and $\Delta$, and identical in the topological phase $|\mu|<2J$, the trivial phase $|\mu|>2J$, and at the critical point $|\mu|=2J$.

The full affine dynamics of $C(t)$ carries an additional eigenvalue $0$, corresponding to the stationary state $C_\infty=\tfrac{1}{2}\mathds{1}_{2L}$; the complete spectrum is therefore $\{0\}\cup\{\lambda_{mn}\}_{m\neq n}$. The gap Eq.~\eqref{eq:gap_result} governs only the non-zero part, i.e., how fast deviations from $C_\infty$ decay. Since $\Delta_{\mathrm{cov}}$ is $\mu$-independent, any $\mu$-dependence of $\tau_{\mathrm{d}}^{\mathrm{g}}$ must originate entirely from the \emph{unitary} part of the dynamics encoded in $e^{-i\mathcal{H}_{\mathrm{BdG}}t}$, not from spectral differences in the dissipator.

\subsection{Perturbative decoherence time at small $\gamma$}
\label{app:perturbative}

We derive the leading-order shift of $\tau_{\mathrm{d}}^{\mathrm{g}}$ away from the unitary maximum $t_*$ when $\gamma$ is small. Throughout, $t_*$ is taken to be the first local maximum of $s(t)=\mathcal{E}_F(t;\gamma=0)$ on $(0,\infty)$; if the unitary trajectory has multiple local maxima, the formula applies at each in turn, and the geometric decoherence time is the shift of the earliest such maximum. The argument rests on four assumptions: (i)~at $\gamma=0$, the pure-state identity gives $\mathcal{E}_F(t;0)=S_{1/2}(t;0)\equiv s(t)$; (ii)~$s(t)$ has a first non-degenerate local maximum at $t_*$, i.e., $s'(t_*)=0$ and $s''(t_*)<0$; (iii)~$S_{1/2}$ and $\mathcal{E}_F$ are differentiable in $\gamma$ near $\gamma=0$ and $t=t_*$; (iv)~the first crossing of the geometric decoherence condition is transverse. Define the first-order response coefficients
\begin{equation}
	s_1(t)
	:= \left.\frac{d}{d\gamma}S_{1/2}(t;\gamma)\right|_{\gamma=0},
	\qquad
	n_1(t)
	:= \left.\frac{d}{d\gamma}\mathcal{E}_F(t;\gamma)\right|_{\gamma=0}.
\end{equation}

\begin{proposition}
	\label{prop:perturbative}
Under assumptions~(i)--(iv) and provided
	$s_1'(t_*)\ge n_1'(t_*)$,
	\begin{equation}
		\tau_{\mathrm{d}}^{\mathrm{g}}
		=
		t_*
		- \gamma\,\frac{n_1'(t_*)}{s''(t_*)}
		+ O(\gamma^2).
		\label{eq:tau_perturbative_app}
	\end{equation}
\end{proposition}

\begin{proof}
	From the exact balanced solution $C_\gamma(t)=\tfrac{1}{2}\mathds{1}_{2L} +e^{-2\gamma t}(C_0(t)-\tfrac{1}{2}\mathds{1}_{2L})$, where $C_0(t)=e^{-i\mathcal{H}_{\mathrm{BdG}}t}C(0) e^{+i\mathcal{H}_{\mathrm{BdG}}t}$, expanding $e^{-2\gamma t}=1-2\gamma t+O(\gamma^2)$ gives
	\begin{equation}
		C_\gamma(t)
		= C_0(t)
		- 2\gamma t\bigl(C_0(t)-\tfrac{1}{2}\mathds{1}_{2L}\bigr)
		+ O(\gamma^2).
	\end{equation}
Under assumption~(iii), both observables inherit smooth $\gamma$-expansions at fixed $t$:
	\begin{equation}
		\begin{aligned}
			S_{1/2}(t;\gamma) &= s(t) + \gamma\,s_1(t) + O(\gamma^2),\\
			\mathcal{E}_F(t;\gamma) &= s(t) + \gamma\,n_1(t) + O(\gamma^2),
		\end{aligned}
	\end{equation}
sharing the same leading term $s(t)$ by assumption~(i).
	
By Definition~\ref{def:geom_dec}, $\tau_{\mathrm{d}}^{\mathrm{g}} =\inf\{t>0:\dot{\mathcal{E}}_F(t)<0 \text{ while }\dot{S}_{1/2}(t)\ge 0\}$. Writing $t=t_*+\delta$ with $\delta=O(\gamma)$ and expanding $\dot{\mathcal{E}}_F$ to first order in $\delta$ and $\gamma$,
	\begin{equation}
		\dot{\mathcal{E}}_F(t_*+\delta;\gamma)
		=
		s''(t_*)\,\delta + \gamma\,n_1'(t_*)
		+ O(\delta^2,\gamma\delta,\gamma^2).
	\end{equation}
The first term arises because $\dot{s}(t_*)=0$, so the leading Taylor correction to $\dot{s}$ is $s''(t_*)\delta$; the second term is the time derivative of the $\gamma$-linear part of $\mathcal{E}_F$ evaluated at $t_*$. Setting $\dot{\mathcal{E}}_F=0$ and solving for $\delta$ using $s''(t_*)<0$,
	\begin{equation}
		\delta = -\gamma\,\frac{n_1'(t_*)}{s''(t_*)},
	\end{equation}
which is Eq.~\eqref{eq:tau_perturbative_app}. To confirm the consistency condition, expand $\dot{S}_{1/2}$ at the same point:
	\begin{equation}
		\dot{S}_{1/2}(t_*+\delta;\gamma)
		= s''(t_*)\,\delta + \gamma\,s_1'(t_*)
		+ O(\delta^2,\gamma\delta,\gamma^2).
	\end{equation}
Substituting $\delta=-\gamma n_1'(t_*)/s''(t_*)$,
	\begin{equation}
		\dot{S}_{1/2}
		= \gamma\bigl[s_1'(t_*)-n_1'(t_*)\bigr] + O(\gamma^2),
	\end{equation}
which is non-negative to leading order precisely when $s_1'(t_*)\ge n_1'(t_*)$, as assumed.
\end{proof}

Equation~\eqref{eq:tau_perturbative_app} admits a simple physical interpretation. The balanced dissipator multiplies the deviation $C_\gamma(t)-\tfrac{1}{2} \mathds{1}_{2L}$ by the strictly decreasing factor $e^{-2\gamma t}$, so  $n_1(t)<0$ throughout. Moreover, since this suppression envelope grows in magnitude monotonically with $t$, the rate of suppression at $t_*$ is directed toward more negative values, giving $n_1'(t_*)<0$. Combined with $s''(t_*)<0$ this gives $n_1'(t_*)/s''(t_*)>0$, so the shift $\delta=-\gamma n_1'(t_*)/s''(t_*)$ is negative: $\tau_{\mathrm{d}}^{\mathrm{g}}$ decreases from $t_*$ as $\gamma$ increases, with the magnitude of the slope $n_1'(t_*)/s''(t_*)$ encoding the $\mu$-dependent sensitivity. Both $s''(t_*)$ and $n_1'(t_*)$ are purely unitary quantities determined by the $\gamma=0$ trajectory $C_0(t)$, and $\gamma$ enters only as the overall prefactor of the shift. The maximum $t_*$ itself is set by the finite-chain quasiparticle traversal time, of order $L/v_{\max}$. Near the phase boundary $|\mu|\to 2J$, the bulk gap closes, and the perturbative expansion can become parametrically delicate if the curvature $|s''(t_*)|$ becomes small or if the first-order response $n_1'(t_*)$ grows. In that regime even a modest $\gamma$ can produce a large shift $\delta$, so the expansion may break down unless the ratio $n_1'(t_*)/s''(t_*)$ remains controlled. The consistency condition $s_1'(t_*)\ge n_1'(t_*)$ is verified numerically to hold throughout the balanced topological-phase scan of Sec.~\ref{sec:single_body}.

\section{Fixed-$\gamma$ cuts near the chiral region}
\label{app:gaussian_fixed_gamma_mu}
To assess the residual $\mu$ dependence at fixed dissipation, Fig.~\ref{fig:appendix_gaussian_fixed_gamma_mu} compares two balanced gain/loss cuts at $\gamma=0.15$. The topological reference point $(\mu,\gamma)=(0.50,0.15)$ and the chiral point $(\mu,\gamma)=(0.00,0.15)$ show very similar behavior: in both cases $\tau_{\mathrm{d}}^{\mathrm{peak}}$ lies slightly below $\tau_{\mathrm{d}}^{\mathrm{g}}$, with only a modest quantitative difference between the two cuts. This comparison shows that, at fixed moderate dissipation, varying $\mu$ does not produce a significant qualitative change, and the peak-based estimate remains a viable proxy for the geometric criterion, albeit with a slight quantitative undershoot. The stronger failure mode is instead the weak-dissipation overshoot inside the topological phase discussed in the main text; see Fig.~\ref{fig:gaussian_fixed_mu_gamma_compare}. In both panels, the geometric decoherence time $\tau_{\mathrm{d}}^{\mathrm{g}}$ is extracted in the background from Definition~\ref{def:geom_dec} by applying the geometric criterion to the full trajectory $\Gamma(t)=\bigl(S_{1/2}(t),\mathcal{N}(t)\bigr)$.
\begin{figure}[htbp]
	\centering
	\begin{subfigure}[t]{0.48\textwidth}
		\centering
		\includegraphics[width=\linewidth]{gaussian_fixed_mu_0p50_gamma_0p15.png}
		\caption{}
	\end{subfigure}\hfill
	\begin{subfigure}[t]{0.48\textwidth}
		\centering
		\includegraphics[width=\linewidth]{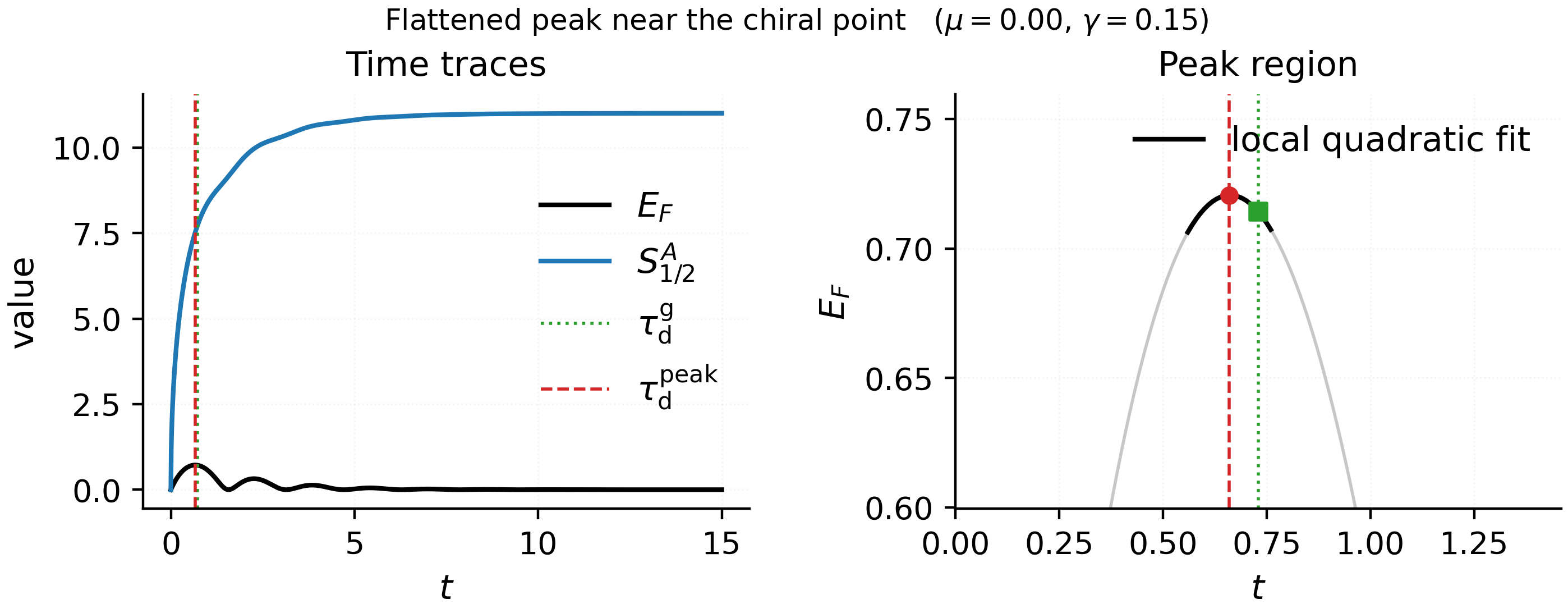}
		\caption{}
	\end{subfigure}
	\caption{Fixed-$\gamma$ comparison at $\gamma=0.15$. Left: topological reference point, $\mu=0.50$. Right: chiral point, $\mu=0.00$. Both panels show only a slight undershoot of $\tau_{\mathrm{d}}^{\mathrm{peak}}$ relative to $\tau_{\mathrm{d}}^{\mathrm{g}}$. Although the entropy--negativity trajectory is not shown, $\tau_{\mathrm{d}}^{\mathrm{g}}$ is extracted from Definition~\ref{def:geom_dec} using the full trajectory $\Gamma(t)=\bigl(S_{1/2}(t),\mathcal{N}(t)\bigr)$.}
	\label{fig:appendix_gaussian_fixed_gamma_mu}
\end{figure}

\section{Finite-Size and Subsystem Scaling of Gaussian Diagnostics}
\label{app:finite-size}

To assess the robustness of the extracted dynamical scales, we study both finite-size effects and dependence on the boundary-block size. Figure~\ref{fig:gaussian_scaling} shows that both $\tau_{\mathrm{d}}^{\mathrm{g}}$ and the peak fermionic Gaussian negativity remain flat across the ranges of $L$ and $L_A$ considered, indicating convergence for the chosen parameters.

\begin{figure}[H]
	\centering
	\includegraphics[width=\columnwidth]{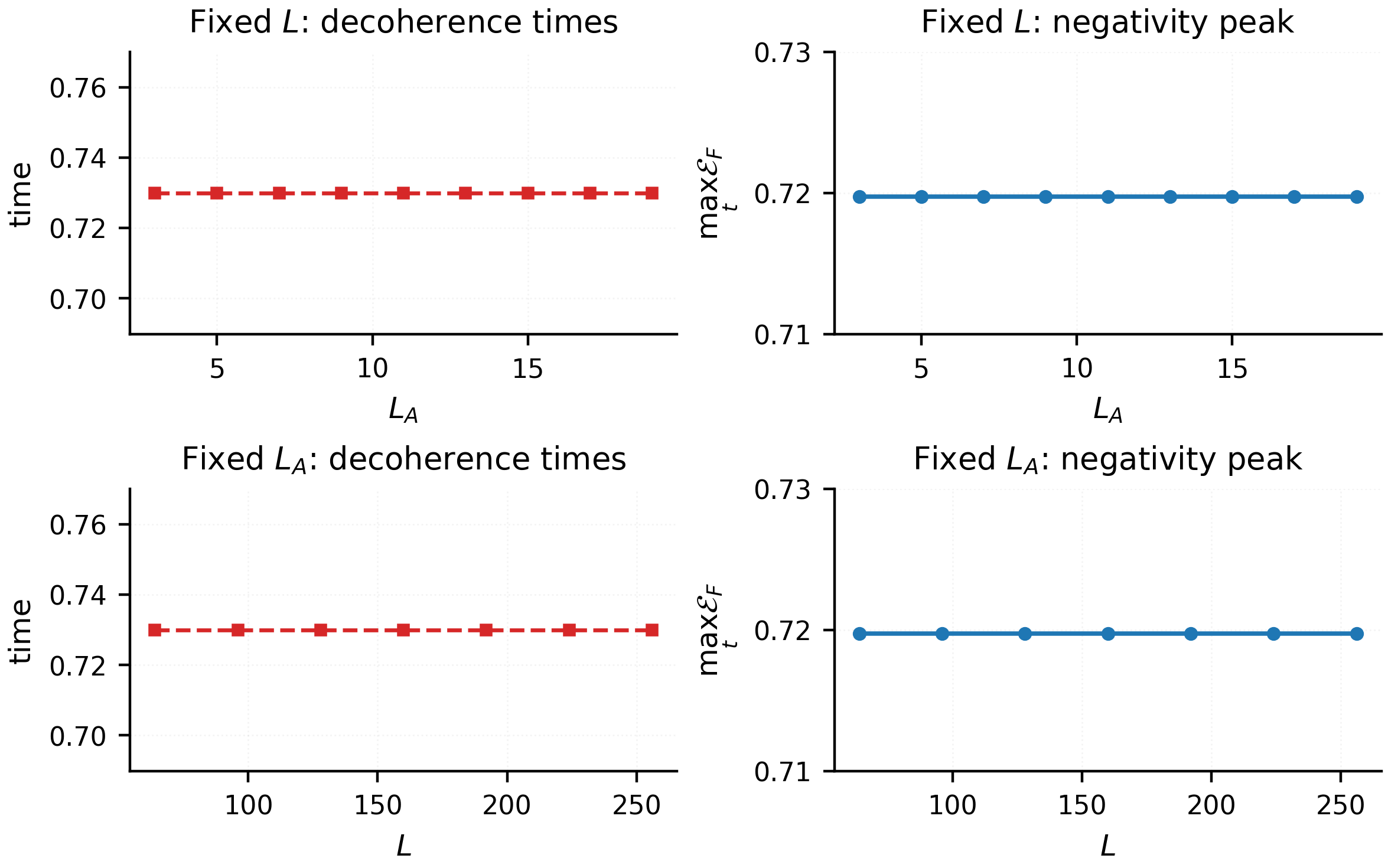}
	\caption{Finite-size and finite-cut dependence of the Gaussian nonequilibrium diagnostics. Left column: geometric decoherence time $\tau_{\mathrm{d}}^{\mathrm{g}}$. Right column: peak fermionic Gaussian negativity $\max_t\mathcal{E}_F$. Top row: dependence on boundary-block size $L_A$ at fixed $L=128$. Bottom row: dependence on total system size $L$ at fixed $L_A=16$. Both quantities remain flat across the ranges studied, confirming that the geometric decoherence scale and the peak entanglement have converged over the system sizes and bipartitions shown. Parameters: $J=\Delta=1$, $\mu=0.5$, $\gamma_+=\gamma_-=0.15$; top row: $L=128$, $L_A\in\{3,5,7,9,11,13,15,17,19\}$; bottom row: $L_A=16$, $L\in\{64,96,128,160,192,224,256\}$; all quantities extracted from $t\in[0,15]$ with step $0.01$.}
	\label{fig:gaussian_scaling}
\end{figure}

\section{Benchmarks}
\label{app:benchmarks}

This appendix collects three technical benchmarks supporting the Gaussian analysis of Sec.~\ref{sec:single_body}. The first compares the balanced-gain Gaussian correlation-matrix evolution against exact many-body Lindblad evolution at small system size, where both approaches are computationally feasible. The second establishes the distinction between the balanced case, where the closed-form solution Eq.~\eqref{eq:balanced_closed_form} applies, and the imbalanced case, where one must integrate Eq.~\eqref{eq:correlation_matrix_ode}. The third verifies that the mutual-information tracking timescales are converged with respect to both the output time step and the choice of bipartition.

\subsection{Balanced Gaussian dynamics versus exact many-body evolution}

For sufficiently small chains, the full Lindblad equation can be solved in the many-body Hilbert space, providing a direct comparison against the Gaussian correlation-matrix treatment. Figure~\ref{fig:appendix_exact_vs_gaussian} shows this comparison at the level of subsystem entropies: the time dependence of $S_{1/2}$, $S_2$, and $S_{\mathrm{vN}}$ obtained from the balanced closed-form Gaussian evolution is compared against those from exact many-body Lindblad evolution implemented in QuTiP~\cite{Johansson2012Aug}. The agreement over the full time window validates the Gaussian framework at small system size.

\begin{figure*}[htbp]
	\centering
	\includegraphics[width=\textwidth]{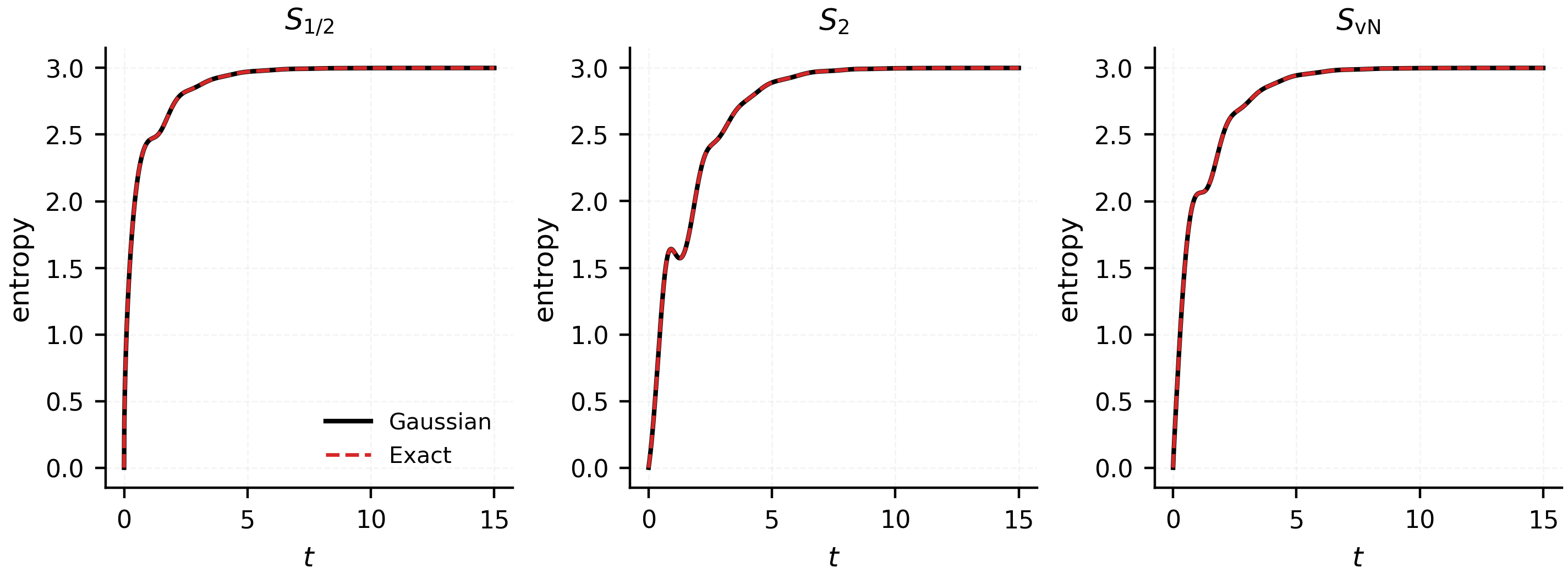}
	\caption{Small-system benchmark of the balanced Gaussian correlation-matrix solution against exact many-body Lindblad evolution. The three panels compare $S_{1/2}$, $S_2$, and $S_{\mathrm{vN}}$ from the two approaches, showing excellent agreement over the full time window. Parameters: open Kitaev chain, $L=10$, $L_A=3$, $J=\Delta=1$, $\mu=0.5$, N\'eel initial state, $\gamma_+=\gamma_-=0.15$, $t\in[0,15]$ with step $0.01$.}
	\label{fig:appendix_exact_vs_gaussian}
\end{figure*}

\subsection{Balanced versus imbalanced gain/loss}

When $\gamma_+=\gamma_-$, the correlation-matrix equation of motion Eq.~\eqref{eq:correlation_matrix_ode} admits the closed-form solution (CMS) of Eq.~\eqref{eq:balanced_closed_form}, and the time evolution is obtained spectrally after a single diagonalization of the BdG matrix. When $\gamma_+\neq\gamma_-$, the dynamics remains Gaussian but no analytic simplification of this form is available; one must instead integrate Eq.~\eqref{eq:correlation_matrix_ode} using a high-order ODE solver---we refer to this as the correlation-matrix differential-equation (CMDE) approach. Figure~\ref{fig:appendix_balanced_vs_imbalanced} makes the distinction explicit by plotting the relative Frobenius error between the CMDE solution and the closed-form CMS expression. The error remains at numerical precision in the balanced case, confirming the validity of the closed-form solution, and grows visibly once the same expression is applied in the imbalanced regime.

\begin{figure}[htbp]
	\centering
	\includegraphics[width=\columnwidth]{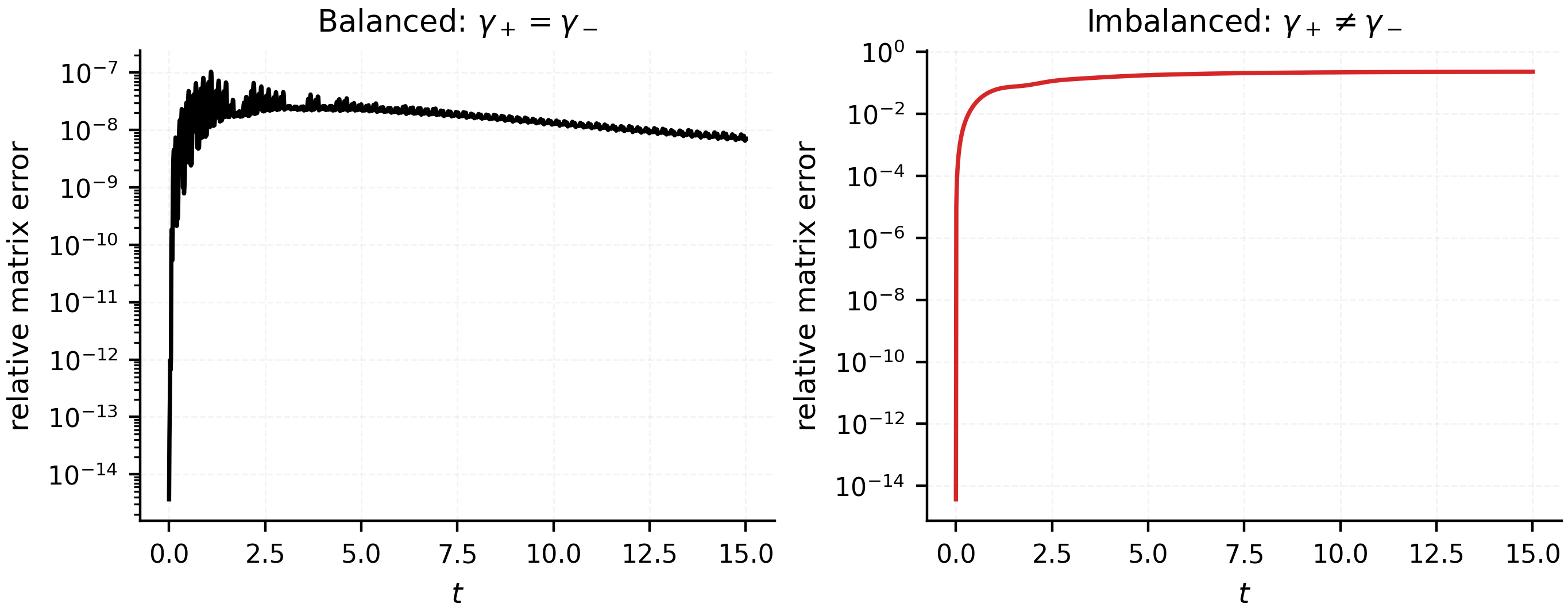}
	\caption{Relative Frobenius error $\|C_{\mathrm{CMDE}}(t)-C_{\mathrm{CMS}}(t)\|_F /\|C_{\mathrm{CMDE}}(t)\|_F$ between the numerical CMDE solution and the closed-form CMS expression Eq.~\eqref{eq:balanced_closed_form}. For $\gamma_+=\gamma_-$ the error stays at numerical precision; for $\gamma_+\neq\gamma_-$ it grows, confirming that the closed-form solution is specific to the balanced case and that the full differential equation must be integrated in the imbalanced regime. Parameters: open Kitaev chain, $L=128$, $J=\Delta=1$, $\mu=0.5$, N\'eel initial state; balanced case $(\gamma_+,\gamma_-)=(0.15,0.15)$, imbalanced case $(\gamma_+,\gamma_-)=(0.20,0.10)$, $t\in[0,15]$ with step $0.01$.}
	\label{fig:appendix_balanced_vs_imbalanced}
\end{figure}

\subsection{Convergence of the mutual-information tracking}

We verify the robustness of the tracking results shown in Fig.~\ref{fig:gaussian_tracking} against both time-step resolution and bipartition choice. For the baseline parameters $L=128$, $L_A=11$, $\mu=0.5$, $J=\Delta=1$, and $\gamma_+=\gamma_-=0.15$, the geometric decoherence time is $\tau_{\mathrm{d}}^{\mathrm{g}}=0.73$ and the peak-based estimate is $\tau_{\mathrm{d}}^{\mathrm{peak}}=0.66$. The peak fermionic negativity is $\mathcal{E}_F^{\mathrm{peak}}=0.72$ and the peak mutual information is $I^{\mathrm{peak}}=1.07$. Reducing the time step from $0.01$ to $0.005$ (Fig.~\ref{fig:appendix_tracking_dt005}) leaves all four quantities unchanged to high accuracy, confirming that both the geometric decoherence time and the shared peak and decay timescales of $\mathcal{E}_F$ and $I(A{:}B)$---the content of Theorem~\ref{cor:product_decay} and the empirical tracking claim of Sec.~\ref{subsec:MI_asymptotics}---are converged with respect to the sampling interval.

\begin{figure}[htbp]
	\centering
	\includegraphics[width=0.82\columnwidth]{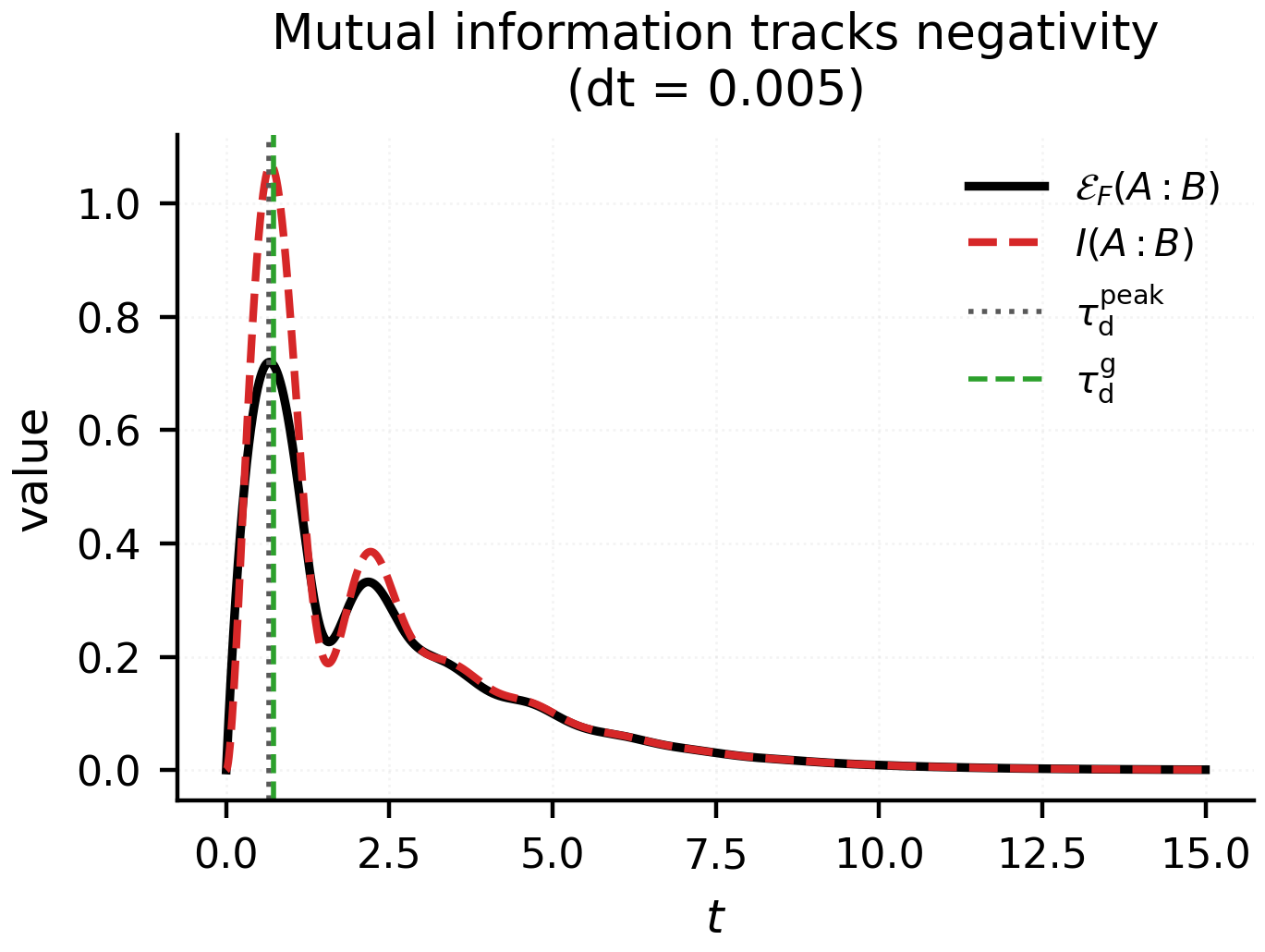}
	\caption{Same diagnostic as Fig.~\ref{fig:gaussian_tracking}, but with time step $\Delta t=0.005$. The mutual information and fermionic Gaussian negativity exhibit the same early-time growth, peak, and subsequent decay as at $\Delta t=0.01$, confirming convergence of the shared timescales with respect to the sampling interval. Parameters: open Kitaev chain, $L=128$, $L_A=11$, $J=\Delta=1$, $\mu=0.5$, N\'eel initial state, $\gamma_+=\gamma_-=0.15$, $t\in[0,15]$.}
	\label{fig:appendix_tracking_dt005}
\end{figure}

Increasing the subsystem from the boundary block $L_A=11$ to the half-chain cut $L_A=64$ (Fig.~\ref{fig:appendix_tracking_halfchain}) leaves $\tau_{\mathrm{d}}^{\mathrm{g}}=0.73$, $\tau_{\mathrm{d}}^{\mathrm{peak}}=0.66$, $\mathcal{E}_F^{\mathrm{peak}}=0.72$, and $I^{\mathrm{peak}}=1.07$ unchanged to high accuracy. This insensitivity to bipartition is consistent with the uniform, site-local nature of the balanced gain/loss protocol: the dominant decoherence scale is set by the global relaxation toward the maximally mixed Gaussian fixed point rather than by any feature specific to a particular cut.

\begin{figure}[htbp]
	\centering
	\includegraphics[width=0.82\columnwidth]{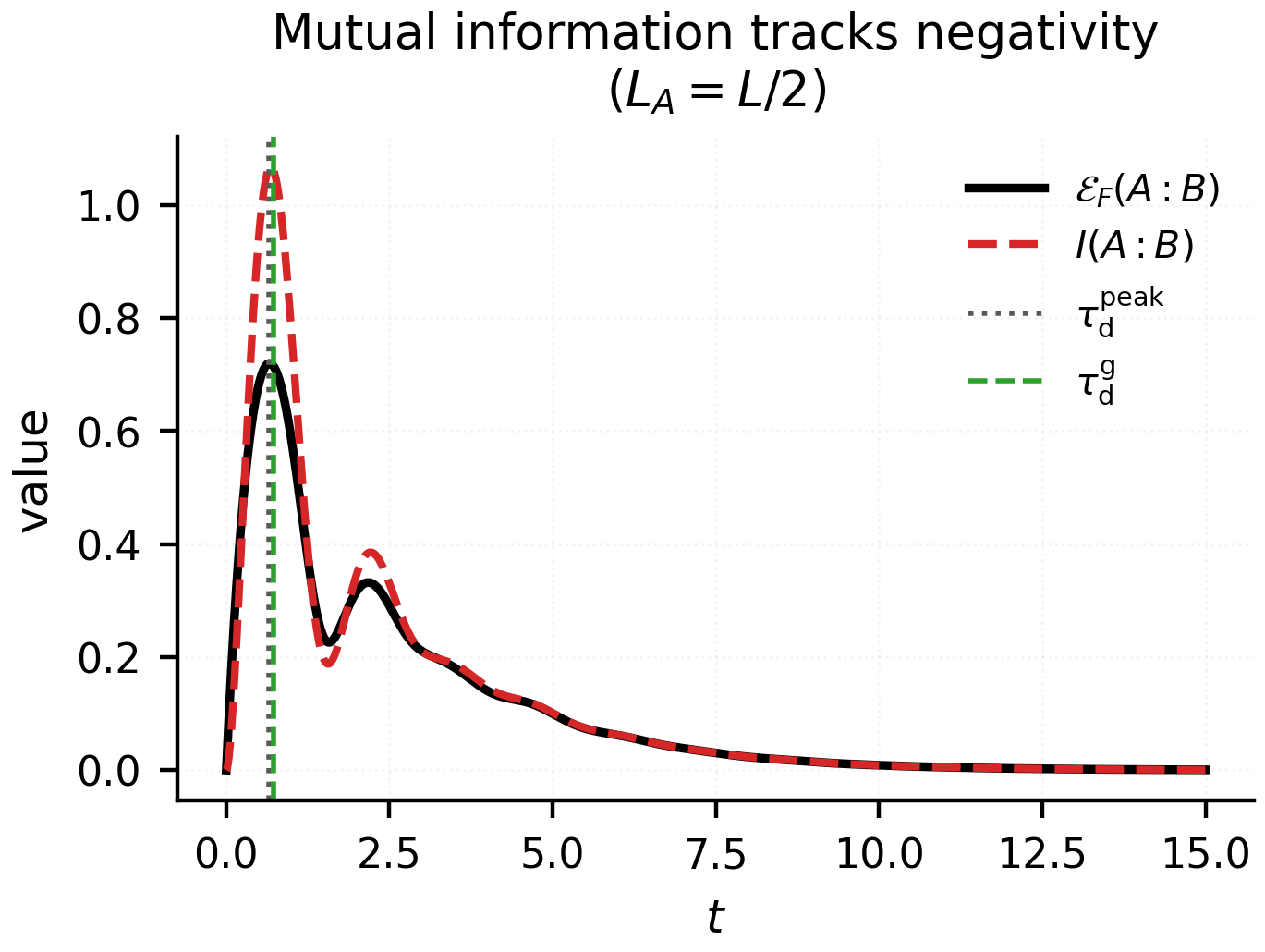}
	\caption{Same diagnostic as Fig.~\ref{fig:gaussian_tracking}, but for a half-chain cut $L_A=L/2=64$. The peak time and asymptotic decay rate of both $\mathcal{E}_F$ and $I(A{:}B)$ are unchanged from the boundary-block result, showing that the tracking relation holds independently of bipartition choice for these parameters. Parameters: open Kitaev chain, $L=128$, $L_A=64$, $J=\Delta=1$, $\mu=0.5$, N\'eel initial state, $\gamma_+=\gamma_-=0.15$, $t\in[0,15]$ with step $0.01$.}
	\label{fig:appendix_tracking_halfchain}
\end{figure}

\section{Finite-Size Robustness of the XXZ Entropy- Negativity Trajectories}
\label{app:xxz-parametric-finite-size}

To verify that the geometric mechanism identified in the interacting XXZ chain is not specific to the $L=10$ system of Sec.~\ref{sec:many_body}, we repeat the entropy--negativity trajectory analysis for $L=8$ and $L=9$. Figures~\ref{fig:xxz_parametric_L8} and~\ref{fig:xxz_parametric_L9} use the same plotting conventions as Fig.~\ref{fig:many_body_parametric}. In both cases, the unitary dynamics pins the trajectory to the pure-state diagonal at early times, and the open-system trajectory bends away at finite time---marking the loss of monotonicity that defines $\tau_{\mathrm{d}}^{\mathrm{g}}$---with the same qualitative behavior as at $L=10$. This confirms that the geometric decoherence mechanism remains robust under modest changes in system size over the accessible exact many-body window. A thorough system-size analysis in the interacting regime will be conducted in future work.

\begin{figure*}[htbp]
	\centering
	\begin{subfigure}[t]{0.48\textwidth}
		\centering
		\includegraphics[width=\linewidth]{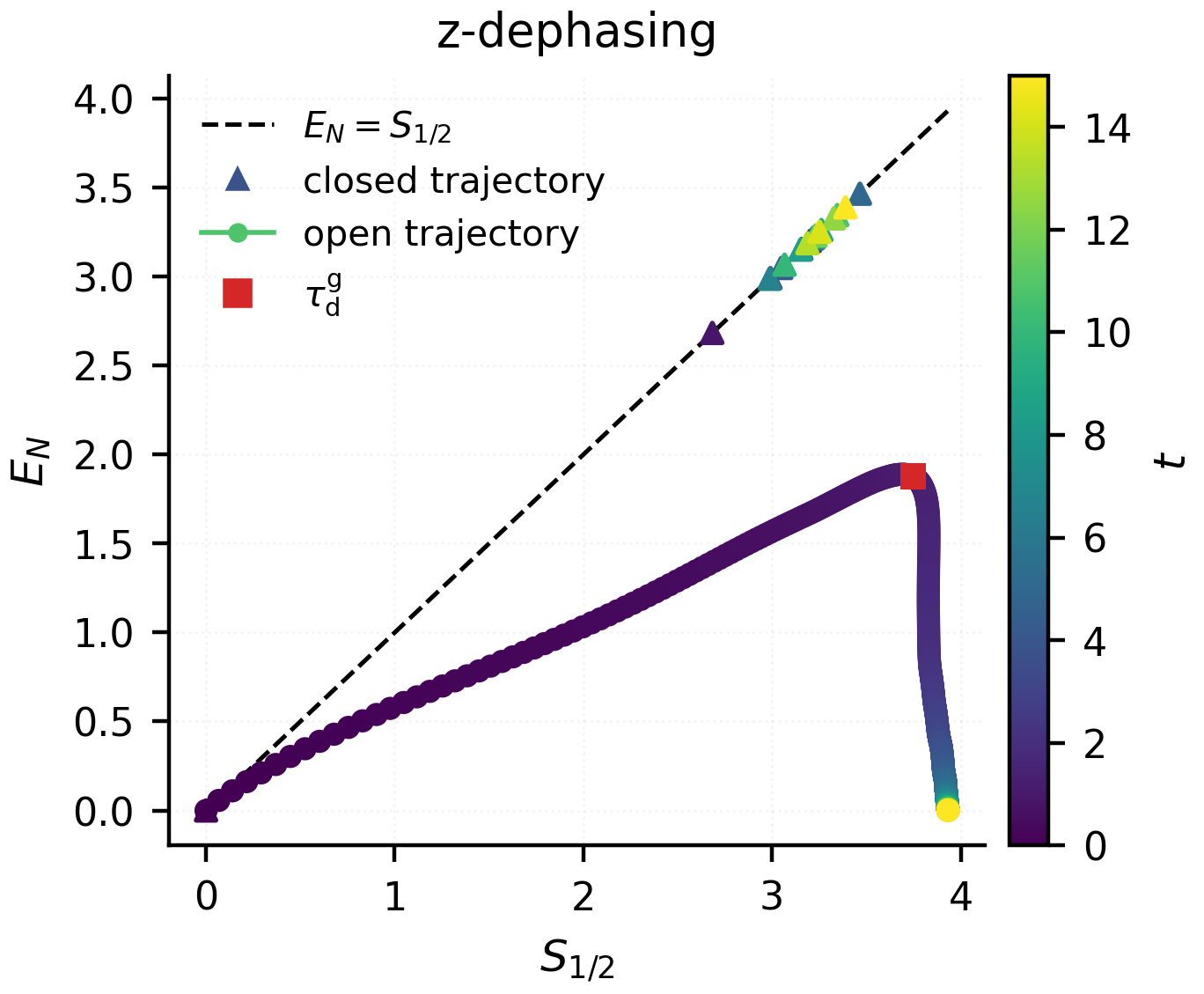}
		\caption{}
	\end{subfigure}\hfill
	\begin{subfigure}[t]{0.48\textwidth}
		\centering
		\includegraphics[width=\linewidth]{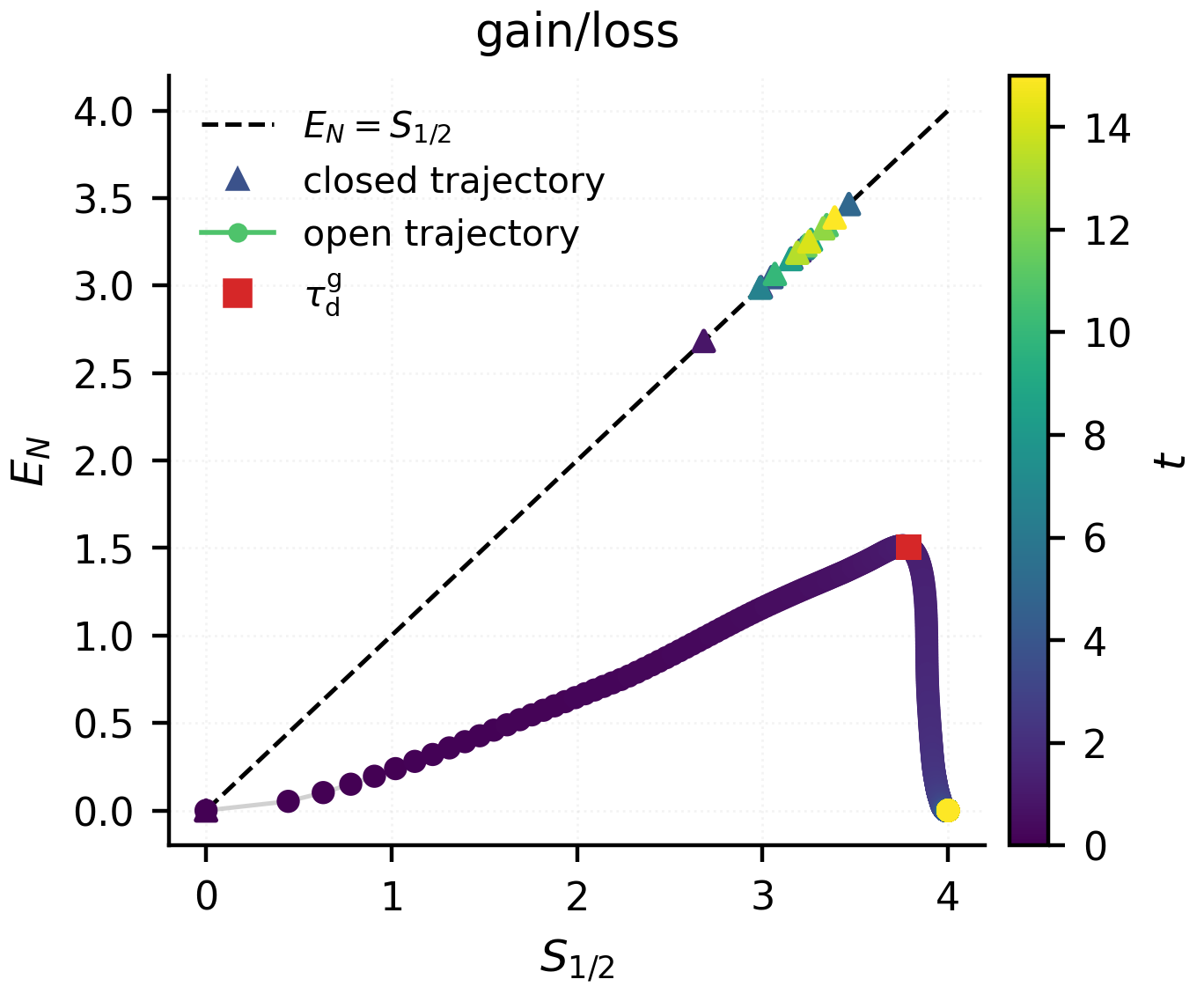}
		\caption{}
	\end{subfigure}
	\caption{Entropy--negativity trajectories for the open XXZ chain with $L=8$ and half-chain bipartition $L_A=4$, using the same conventions as Fig.~\ref{fig:many_body_parametric}. (a)~Local $z$-dephasing. (b)~Balanced gain/loss. The geometric bending that defines $\tau_{\mathrm{d}}^{\mathrm{g}}$ is clearly visible in both panels. Parameters: $J=1$, $J_z=0.55$, N\'eel initial state, $t\in[0,15]$ with step $0.01$; panel~(a) uses $\gamma_z=0.15$, panel~(b) uses $\gamma_+=\gamma_-=0.15$.}
	\label{fig:xxz_parametric_L8}
\end{figure*}

\begin{figure*}[htbp]
	\centering
	\begin{subfigure}[t]{0.48\textwidth}
		\centering
		\includegraphics[width=\linewidth]{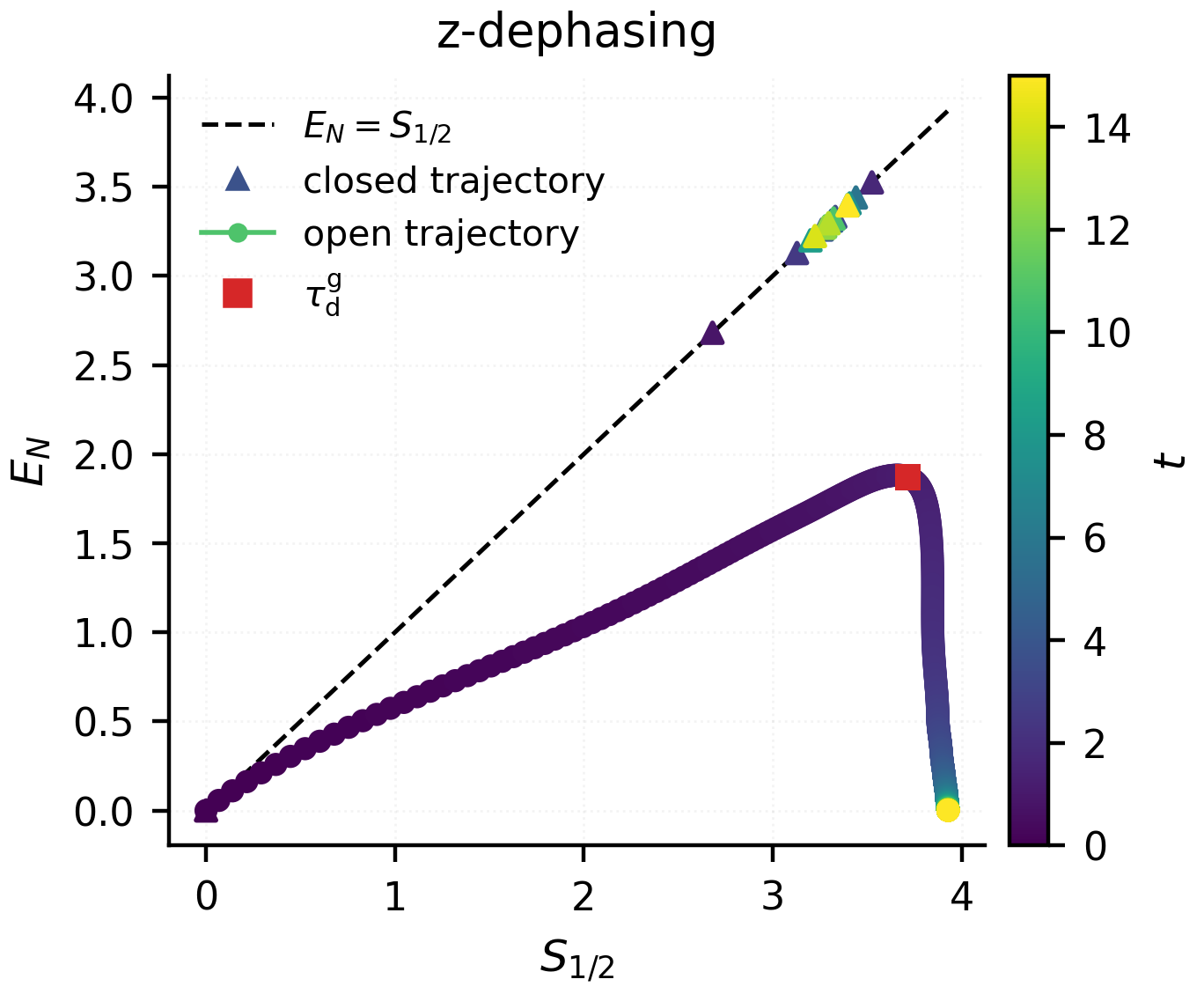}
		\caption{}
	\end{subfigure}\hfill
	\begin{subfigure}[t]{0.48\textwidth}
		\centering
		\includegraphics[width=\linewidth]{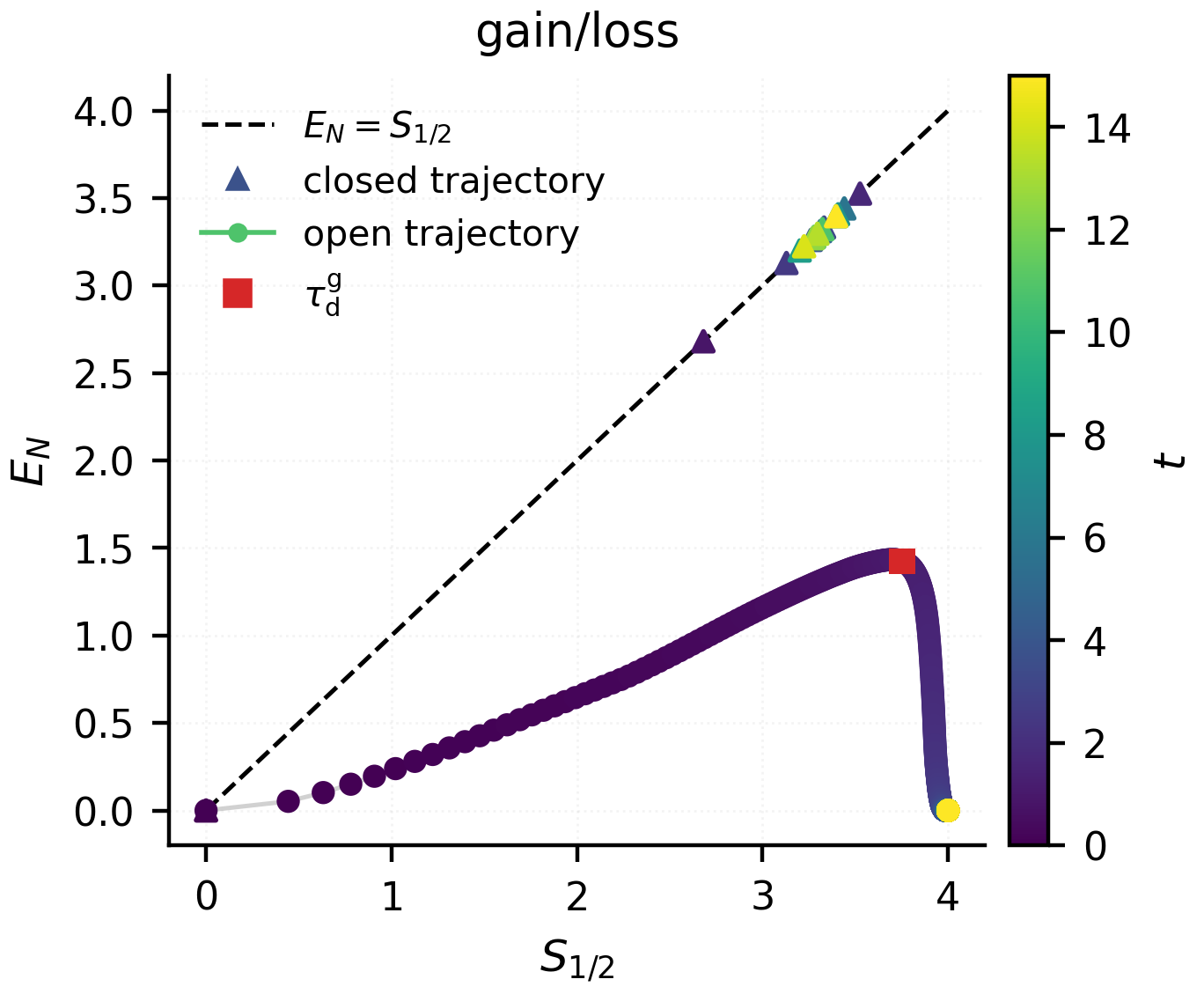}
		\caption{}
	\end{subfigure}
	\caption{Entropy--negativity trajectories for the open XXZ chain with $L=9$ and left-block bipartition $L_A=4$, using the same conventions as Fig.~\ref{fig:many_body_parametric}. (a)~Local $z$-dephasing. (b)~Balanced gain/loss. As at $L=10$, both panels show a clear loss of monotonicity in the $(S_{1/2},E_N)$ plane at finite time, consistent with the geometric decoherence criterion of Definition~\ref{def:geom_dec}. Parameters: $J=1$, $J_z=0.55$, N\'eel initial state, $t\in[0,15]$ with step $0.01$; panel~(a) uses $\gamma_z=0.15$, panel~(b) uses $\gamma_+=\gamma_-=0.15$.}
	\label{fig:xxz_parametric_L9}
\end{figure*}

\end{document}